\documentclass[pra,twocolumn,nofootinbib,superscriptaddress]{revtex4-1}
\pdfoutput=1
\usepackage{color,amsthm,amsmath,amsxtra,amsfonts,dsfont,graphicx,bm,enumerate}

\def\Id{{\openone}}
\newcommand{\be}{\begin{equation}}
\newcommand{\ee}{\end{equation}}
\newcommand{\bea}{\begin{eqnarray}}
\newcommand{\eea}{\end{eqnarray}}

\newtheorem{thm}{Theorem}[section]
\newtheorem{prop}[thm]{Proposition}
\newtheorem{cor}[thm]{Corollary}
\newtheorem{lem}[thm]{Lemma}
\newtheorem{defn}[thm]{Definition}
\newtheorem{example}[thm]{Example}

\newcommand{\ket}[1]{\vert#1\rangle}
\newcommand{\bra}[1]{\langle#1\vert}
\newcommand{\tr}{\mathrm{tr}}

\begin{document}

\title{Matrix Product Unitaries: Structure, Symmetries, and Topological Invariants}

\author{J.~Ignacio \surname{Cirac}}
\affiliation{Max-Planck-Institut f{\"{u}}r Quantenoptik,
Hans-Kopfermann-Str.\ 1, 85748 Garching, Germany}

\author{David Perez-Garcia}
\affiliation{Departamento de Análisis Matemático,
Universidad Complutense de Madrid, Plaza de Ciencias 3, 28040 Madrid, Spain}
\affiliation{ICMAT, Nicolas Cabrera, Campus de Cantoblanco, 28049 Madrid, Spain}

\author{Norbert Schuch}
\affiliation{Max-Planck-Institut f{\"{u}}r Quantenoptik,
Hans-Kopfermann-Str.\ 1, 85748 Garching, Germany}

\author{Frank Verstraete}
\affiliation{Department of Physics and Astronomy,
Ghent University, Krijgslaan 281, S9, 9000 Gent, Belgium}
\affiliation{Vienna Center for Quantum Science and Technology,
Faculty of Physics, University of Vienna, Boltzmanngasse 5, 1090 Vienna, Austria}

\begin{abstract}
Matrix Product Vectors form the appropriate framework to study and
classify one-dimensional quantum systems.  In this work, we develop the
structure theory of Matrix Product Unitary operators (MPUs) which appear
e.g.\ in the description of time evolutions of one-dimensional systems. We
prove that all MPUs have a strict causal cone, making them Quantum
Cellular Automata (QCAs), and derive a canonical form for MPUs which
relates different MPU representations of the same unitary through a local
gauge. We use this canonical form to prove an Index Theorem for MPUs which
gives the precise conditions under which two MPUs are adiabatically
connected, providing an alternative derivation to that of
Ref.~\cite{gross:index-theorem} for QCAs.  We also discuss the effect of
symmetries on the MPU classification. In particular, we characterize the
tensors corresponding to MPU that are invariant under conjugation,
time reversal, or transposition. In the first case, we give a full
characterization of all equivalence classes. Finally, we give several
examples of MPU possessing different symmetries.
\end{abstract}

\maketitle

% ============================================================
\section{Introduction}

Strongly correlated quantum systems display a wide range of unconventional
phenomena which emerge from the close interplay between their locality
structure and the presence of strong interactions.  Tensor Networks
accurately capture the complex correlations present in such systems while
at the same time preserving their locality structure, and thus form a
versatile framework for their analytical characterization and numerical
study. In one spatial dimension, this gives rise to a Matrix Product
Vector (MPV) structure which appears in a wide range of scenarios,
encompassing both genuinely one-dimensional systems and boundaries of
two-dimensional systems.  Specifically, Matrix Product States (MPS) allow
for the efficient description of low-energy states of one-dimensional
systems~\cite{hastings:arealaw,verstraete:faithfully}. Matrix Product
Density Operators
(MPDOs)~\cite{verstraete:finite-t-mps,zwolak:mixedstate-timeevol-mpdo}
describe thermal states in one dimension as well as boundaries of gapped
two-dimensional systems~\cite{cirac:peps-boundaries}, whereas topological
order and anyonic excitations can be characterized using Matrix Product
Operator Algebras~\cite{sahinoglu:mpo-injectivity,bultinck:mpo-anyons}.
The evolution under one-dimensional Hamiltonians and Floquet
operators~\cite{po:chiral-floquet-mbl-etal}, as well as the effect of bulk
symmetries on boundaries of two-dimensional
systems~\cite{chen:2d-spt-phases-peps-ghz}, can be described in terms of
Matrix Product Unitaries (MPUs).

A central application of MPVs is the study and classification of phases of
strongly correlated systems, this is, equivalence classes under
smooth deformations, based on discrete invariants labelling the MPVs of
interest.  The precise notion of equivalence depends on the
physical context and the constraints imposed, where overall
two scenarios need to be distinguished: For MPVs describing
genuinely one-dimensional systems, the appropriate notion of equivalence
characterizes whether it is possible to smoothly deform an MPV into
another one while keeping certain properties.  In contrast, for MPV
characterizing boundaries of two-dimensional systems a smooth
interpolation of the two-dimensional bulk is required, to which end it
rather matters whether one can locally combine the boundary MPVs into a
joint MPV with the same structure.

One of the most important achievements of MPS has been the full
classification of zero-temperature phases of gapped one-dimensional
quantum systems, based on previously derived canonical forms of MPS and
their connection to parent
Hamiltonians~\cite{perez-garcia:mps-reps}.  It has been proven that
without symmetries, any two MPS can be connected along a path
corresponding to a gapped Hamiltonian, while in the presence of
symmetries, the inequivalent phases are characterized by the symmetry
action at the boundary (this is, on the entanglement), labelled by group
cohomology and generalizations
thereof~\cite{pollmann:1d-sym-protection-prb,chen:1d-phases-rg,schuch:mps-phases}.
Invariants characterizing different phases have also been identified in a
range of other contexts: Local symmetries of 2D phases are represented at
the boundary by MPUs, which in turn are labelled by the third group
cohomology~\cite{chen:2d-spt-phases-peps-ghz}, MPOs describing
topologically ordered wavefunctions locally can be characterized in
terms of tensor fusion
categories~\cite{sahinoglu:mpo-injectivity,bultinck:mpo-anyons}, and
scale-invariant MPDOs appearing at the boundary of renormalization fixed
points exhibit an algebraic structure given by fusion
categories~\cite{cirac:mpdo-rgfp}. The
central tool for all these results has been the identification of suitable
canonical forms for the corresponding MPVs, which allow to relate
different MPV descriptions of the same state locally, and to classify
phases through the algebraic structure of the underlying gauge
transforms~\cite{perez-garcia:mps-reps,cirac:mpdo-rgfp}.

In this paper, we develop a structure theory of Matrix Product Unitaries
(MPU) for one-dimensional systems. Specifically, we provide a
characterization of the general structure of arbitrary MPUs, and use this
to derive a canonical form for MPUs which relates different
representations of the same MPU through a local gauge. We find that all
MPUs are, in fact, Quantum Cellular Automata (QCAs)~\cite{schumacher:qca},
this is, they propagate information strictly by a finite distance only.
Based on our standard form, we define an index which measures the net
quantum information flow through the MPU, and give a Matrix Product based
derivation of the Index Theorem of Gross \emph{et
al.}~\cite{gross:index-theorem} which
shows that the index precisely labels equivalence classes of MPUs under
smooth deformations.  We also characterize the tensors generating MPUs
with one of the following symmetries: conjugation, time reversal, and
transposition. For the case of conjugation, we fully classify the
corresponding equivalence classes, and for all symmetries, we
connect our results with the characterization of symmetry protected
topological phases in MPS. We conclude with some relevant examples of MPUs
possessing the different symmetries.

Recently, MPUs have been used to describe the time evolution operator in the
study of Floquet phases, based on the intuition that the Matrix Product
structure captures the time evolution under local
Hamiltonians~\cite{po:chiral-floquet-mbl-etal}; with the
additional assumption that these MPUs describe QCAs, this allowed to use
the Index Theorem of Ref.~\cite{gross:index-theorem} to classify certain
Floquet phases.  Our result removes the need for this additional
assumption by proving its general validity, and gives a direct way of
classifying MPUs based on their standard form, providing the tools for
generalizations to systems with symmetries in the MPV formalism. Let us
note that the effect of symmetries on the classification of time
evolutions has recently been investigated for one-dimensional random walks
of single particles~\cite{cedzich:random-walk-classification}, which in
the absence of symmetries share many common features with QCAs. We believe
that the structure theory of MPUs developed in this work will allow to
generalize these results to the domain of interacting
systems~\cite{fidkowski:interacting-floquet-fermions}.

This paper is organized as follow: In Section II we review some of the
basic properties of MPVs that will be used along the paper. In Section III
we derive the general structure of MPUs and characterize them with a
Standard Form. In Section IV we introduce the
Index~\cite{gross:index-theorem} and prove its robustness. Section V
contains the characterization of tensors corresponding to MPUs with
different symmetries. Section VI analyzes the equivalences among them, and
Section VII includes some relevant examples. The connection between MPUs
and QCAs is given in the Appendix.

% ============================================================
\section{Basics in Matrix Product Vectors}

In this section we will recall the basic definitions and results about
Matrix Product Vectors from Ref.~\cite{cirac:mpdo-rgfp}.

Let us consider a Hilbert space of dimension $d_0$, $H_{d_0}$, a basis $\{e_{n}\}_{n=1}^{d_0}$, and a tensor, ${\cal A}$, of complex coefficients $A_{\alpha,\beta}^n$, with $n=1,\ldots, d_0$ and $\alpha,\beta=1,\ldots,D$. The tensor defines a family of vectors, $\{V^{(N)}\in H_d^{\otimes N}\}_{N=1}^{\infty}$, as
 \be
 V^{(N)} = \sum_{n_1,\ldots,n_N=1}^{d_0} c_{n_1,\ldots,n_N} \; e_{n_1}\otimes\ldots\otimes e_{n_N}
 \ee
where
 \be
 \label{coefficients}
 c_{n_1,n_2,\ldots,n_N} = \sum_{\alpha_1,\ldots, \alpha_N=1}^D
 \prod_{k=1}^N A^{n_k}_{\alpha_k,\alpha_{k+1}}
 \ee
are complex coefficients and with the identification
$\alpha_{N+1}=\alpha_1$. We call the $V^{(N)}$ matrix product vectors
(MPV), and we say that ${\cal A}$ generates the MPV, where the $\alpha$'s
and $n$ are the bond (or auxiliary) and physical indices, respectively,
and $D$ and $d_0$ their corresponding dimensions. We will use the same
graphical notation as in \cite{cirac:mpdo-rgfp}. First, we represent the tensor ${\cal A}$ as
 \be
 {\includegraphics[height=2.8em]{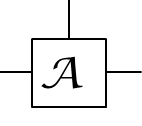}}
\ee
where the horizontal (vertical) lines represent the bond (physical) indices. We will represent the MPV
 \be
 V^{(N)}=\;\raisebox{-34pt}{\includegraphics[height=5.4em]{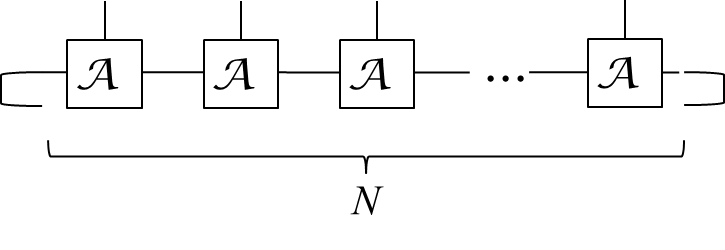}}
 \ee
Here, the lines that join different tensors indicate that the corresponding indices are contracted, while the vertical lines (which are open) represent the physical indices. The curvy lines at the end indicate that the last and first tensors are also contracted.
This graphic means that the coefficients of $V^{(N)}$ can be determined by assigning values to the physical (open) indices, multiplying the matrices and then taking the trace.

\begin{defn}
The {\em transfer operator} is defined as
 \be
 E = \sum_{n=1}^D A^n \otimes \bar A^n
 \ee
where $A^n$ are matrices with elements $A^n_{\alpha,\beta}$, and the bar denotes complex conjugate. We will denote by $\lambda_E$ the spectral radius of $E$, i.e. the eigenvalue with largest absolute value.
\end{defn}

As the MPV generated by two different tensors may coincide, it is useful to define a canonical form:
\begin{defn}
We say that a tensor ${\cal A}$ generating MPV is in {\em canonical form}
(CF) if: (i) the matrices are of the form $A^n = \oplus_{k=1}^r \mu_k A^n_k$, where $\mu_k\in \mathds{C}$ and the spectral radius of the transfer matrix, $E_k$, associated to $A^n_k$ is equal to one; (ii) for all $k$, there exists no projector, $P_k$, such that $A^n_k P_k = P_k A^n_k P_k$ for all $n$, or $P_k A^n_k = P_k A^n_k P_k$ for all $n$.
\end{defn}
This basically means that the matrices $A^n$ are written in a block-diagonal form and there is no way of making the blocks smaller. Furthermore, if two tensors, ${\cal A},{\cal B}$, are related to each other by a gauge transformation (${\cal B}= X{\cal A}X^{-1}$) then they generate the same MPV. Using this fact, one can show that for any tensor it is always possible to find another one that generates the same MPV and is in CF.

\begin{defn}
We say that a tensor generating MPV is {\em normal} if it is in CF, has a single block ($r=1$), and its associated transfer matrix has a unique eigenvalue of magnitude (and value) equal to its spectral radius, which is equal to one.
\end{defn}
Note that the transfer matrix $E_k$ associated to $A^n_k$ in CF may still have several eigenvalues of magnitude equal to their spectral radius (the so-called $p$-periodic states). For normal tensors, however, this is not the case for its single block. Furthermore, its transfer matrix has no Jordan block corresponding to the eigenvalue one. Moreover, the condition that the spectral radius is equal to one can always be met by multiplying the tensor by a constant.

Let us denote by $\Phi$ and $\rho$ the left and right eigenvectors of $E$ corresponding to the eigenvalue 1 of a normal tensor.

\begin{defn}
For normal tensors, we say that ${\cal A}$, generating MPV, is in {\em Canonical Form II} (CFII) if
 \begin{subequations}
 \label{Erightleft}
 \bea
 (\Phi|&=&\sum_{n=1}^D (n,n|,\\
 |\rho)&=&\sum_{n=1}^D \rho_n |n,n),
 \eea
 \end{subequations}
where $\rho_n>0$ and $(\Phi|\rho)=1$.
\end{defn}

Since $\Phi$ and $\rho$ are rank-2 tensors, they can be considered as $D\times D$ matrices, or as vectors in $H_D\otimes H_D$. We will use those considerations depending on the context. Graphically,
 \be
 \label{I_Transfer}
 E=\;\raisebox{-20pt}{\includegraphics[height=4.5em]{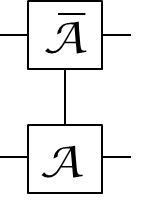}}
 \ee
and
 \be
 {\includegraphics[height=4.3em]{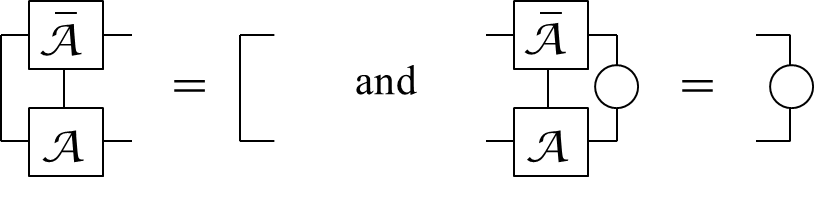}}
 \ee
Here, the tensor with the circle represents the right eigenvector $\rho$.

For a normal tensor it is always possible to find a gauge transformation defining a new normal tensor that is in CFII and generates the same MPV.
Another concept that we will use later on is that of blocking.

\begin{defn}
\label{blocking}
Given ${\cal A}$ generating MPV, we denote by ${\cal A}_k$ the tensor obtained after blocking $k$ tensors
\be
 \label{MPUblock}
 {\includegraphics[height=5.0em]{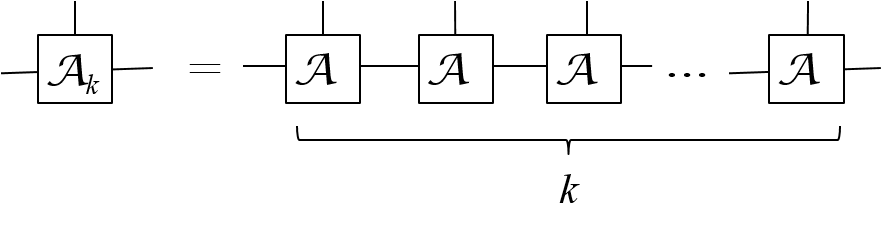}}
\ee
\end{defn}
In (\ref{MPUblock}), the vertical indices have been combined into a single one. Note that $d_k=d_0^k$ and $D$ are the dimension of the local physical Hilbert space and bond dimension corresponding to ${\cal A}_k$, respectively.

In the following, we will transition to Matrix Product Operators, for
which $H_{d_0}$ will describe the space of operators acting on a Hilbert
space $H_d$ of dimension $d$, and the MPV will thus represent operators
acting on $H_d^{\otimes N}$.  In the case where these operators are
unitary, we will term them Matrix Product Unitaries
(MPU)~\cite{chen:2d-spt-phases-peps-ghz}. In order to emphasize that the
tensor generates MPU, we will denote it by ${\cal U}$ and the
corresponding MPU by $U^{(N)}$.

% ====================================================
\section{Stucture of MPU}

The aim of this section is to characterize when a tensor ${\cal U}$ generates translationally invariant MPU, i.e., when
\be
 \label{}
 U^{(N)}=\;\raisebox{-38pt}{\includegraphics[height=5.4em]{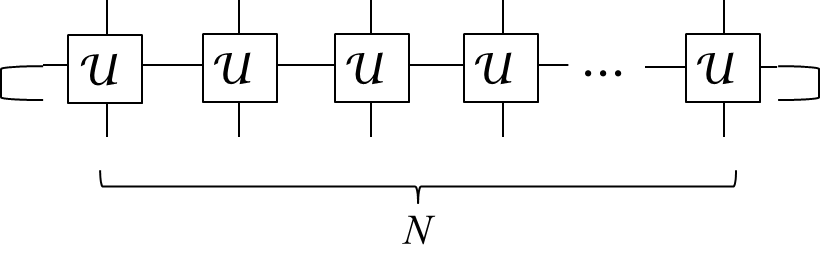}}
\ee
is a unitary operator for $N=2, 3\dots$. We will also introduce a standard form for such tensors and a fundamental theorem which will allow us to relate tensors generating the same family of MPU.

As explained in the previous section, unless we specify the contrary, wlog we can assume that the tensor ${\cal U}$ is in CF. We will denote by  $\bar{\cal U}$ the tensor obtained from ${\cal U}$ by transposing the physical indices and conjugating all coefficients, and use the graphical notation
\be
 \label{}
 {\includegraphics[height=3.0em]{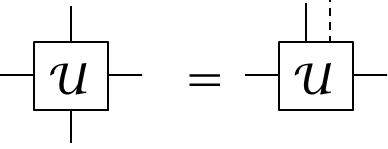}}
\ee
The unitarity of the MPU is equivalent to $U^{(N)\dagger}U^{(N)}=\Id^{\otimes N}$. Graphically,
\be
 \label{UisUnitary}
 {\includegraphics[height=4.6em]{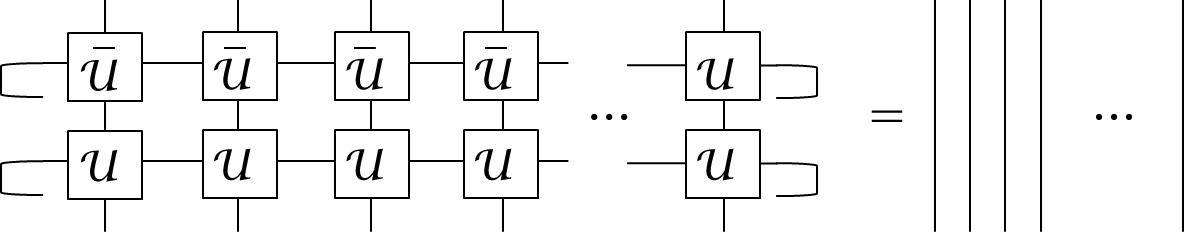}}
\ee
We also define the transfer matrix of $\frac{1}{\sqrt{d}}{\cal U}$
\be
 \label{eq:transfer-op}
 E=\frac{1}{d}\;\;\raisebox{-22pt}{\includegraphics[height=4.2em]{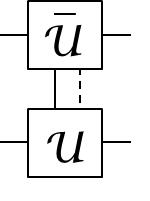}}
\ee

Let us start with the following simple observation:

\begin{prop}
\label{prop:normal-tensor}
$E$ has just one nonzero eigenvalue, which is equal to one, and $\frac{1}{\sqrt{d}}{\cal U}$ is a normal tensor.
\end{prop}

\begin{proof}
Since $U^{(N)}$ is unitary for all $N>1$, then ${\rm
tr}(E^N)=\frac{1}{d^N}{\rm tr}(U^{(N)}U^{(N)\dagger})=1$ for all $N>1$.
Thus, the spectrum of $E$ contains a one and the rest of the eigenvalues
vanish, as it trivially follows from Lemma A.5 in \cite{cirac:mpdo-rgfp}. This is incompatible with having more than one block in the CF. Furthermore, the transfer matrix does not have more than one eigenvalue of magnitude equal to one, so that the tensor must be normal.
\end{proof}

As a consequence of this observation, wlog we can  assume that $\frac{1}{\sqrt{d}}{\cal U}$ is in CFII. Following the Definition \ref{blocking} we will denote by ${\cal U}_k$ the tensor corresponding to blocking ${\cal U}$ $k$ times. Note also that this tensor is also normal and in CFII.

% ============================================================
\subsection{Characterization of MPU via simple tensors}

We will consider a special kind of tensors generating MPU, which have a peculiar property, from which we can derive their general structure. Then we will show that by blocking a final number of times, any tensor generating MPU develops such a property.

\begin{defn}
We say that a tensor ${\cal U}$ is {\em simple} if there exist two tensors, $a$ and $b$, such that
 \begin{subequations}
 \label{simple}
 \bea
 \label{simple1}
  && {\includegraphics[height=5.0em]{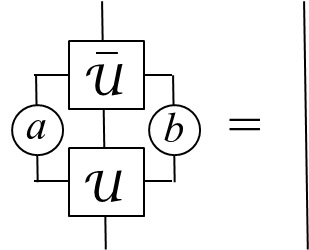}}\\
  && {\includegraphics[height=5.0em]{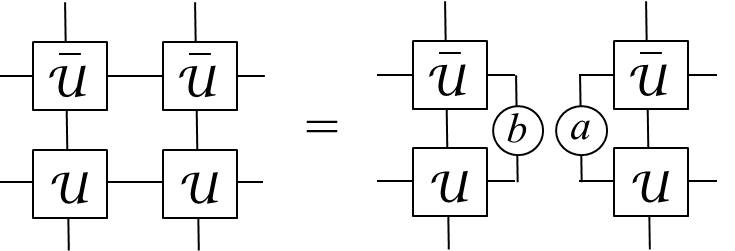}}
 \label{simple2}
 \eea
 \end{subequations}
\end{defn}

Now we show that by blocking, any tensor generating MPU becomes simple.

\begin{prop}
\label{blockingsimple}
(i) Any simple tensor generates MPU. (ii) For any tensor generating MPU,
${\cal U}$, there exists some $k\le D^4$ such that ${\cal U}_k$ is simple,
where tensors $a,b$ are the left and right fixed points, $\Phi$ and
$\rho$,  of the transfer operator $E$, respectively.  
\end{prop}

\begin{proof}
(i) One can just apply sequentially (\ref{simple}) to $U^{(N)\dagger}U^{(N)}$ to obtain
(\ref{UisUnitary}) so that $U^{(N)}$ is indeed unitary for $N>1$.
\\[0.5ex]
(ii) Given ${\cal U}$ generating MPU, we can always write
\be
 \label{WIsom}
 \raisebox{-20pt}{\includegraphics[height=4.6em]{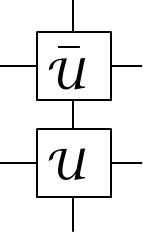}}=E\otimes \Id + \sum_\alpha S_\alpha\otimes \sigma_\alpha
\ee
Here, the first part of the tensor products acts on the auxiliary indices, whereas the second part on the physical ones. Furthermore $E$ is the transfer operator defined in (\ref{eq:transfer-op}) whereas the $\sigma_\alpha$ are traceless operators.

By proposition \ref{prop:normal-tensor}, $E$ has a unique non-zero
eigenvalue equal to $1$. Since ${\cal U}$ is in CFII, its associated  left
and right eigenvectors, $\Phi$ and $\rho$, are given by
(\ref{Erightleft}). Graphically, we have [cf.~(\ref{I_Transfer})]
\be
 \label{II_UTransfer}
 \raisebox{-20pt}{\includegraphics[height=4.4em]{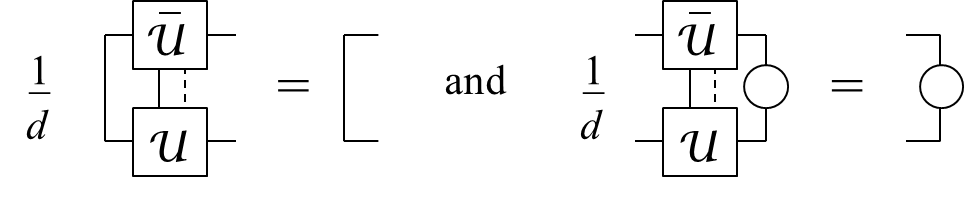}}
\ee
where the circle represents the right eigenvector, $\rho$. Furthermore, there is no Jordan block associated to this eigenvalue, whereas there may be one or several Jordan blocks associated to the zero eigenvalues. If we denote by $J< D^2$ the largest size of all Jordan blocks of $E$, let us block now $J$ sites and consider ${\cal U}_J$. Given that this tensor itself generates MPU and is in CFII, we can use for it the decomposition (\ref{WIsom}), where we will denote by $E'$ and $S'_\alpha$ the new matrices. We have $E'=|\rho)(\Phi|$, since after blocking all Jordan blocks will disappear. Taking into account (\ref{UisUnitary}) we conclude that the trace of any product of operators containing one or more $S'_\alpha$ vanishes. This implies that
 \be
 \label{ESE=0}
 E'S'_{\alpha_1}\ldots S'_{\alpha_m}E'={\rm tr}(E'S'_{\alpha_1}\ldots S'_{\alpha_m})=0
 \ee
for any $m\ge 1$. Furthermore, any element, $S$, in the algebra generated by the $S'_\alpha$ must have zero eigenvalues since ${\rm tr}(S^N)=0$ for all $N$, and thus they are nilpotent. It follows then from a result of Nagata and Higman \cite{Nagata,Higman}, improved later by Razmyslov \cite{Razmyslov}, that there exists some $J'< D^2$ such that
 \be
 \label{Sprime=0}
  S'_{\alpha_1}\ldots S'_{\alpha_{J'}} =0
 \ee
for any set of $\alpha$'s (see \cite{Zelmanov} for a review about this type of questions). Let us now block again, $J'$ times, so that we consider ${\cal U}_{JJ'}$. Again, we can use the decomposition (\ref{WIsom}) and denote by $E''=E'$ and $S''_\alpha$ the new matrices. The latter can have one of the following three types of structures:
 \begin{subequations}
 \label{sprimeforms}
 \bea
 &{\rm type } 1:&\quad S''_\alpha=E'\left[\prod_{k} S'_{\alpha_k}\right],\\
 &{\rm type } 2:&\quad S''_\alpha=\left[\prod_k S'_{\alpha_k}\right] E',\\
 &{\rm type } 3:&\quad S''_\alpha=\left[\prod_{k} S'_{\alpha_k}\right] E' \left[\prod_\ell S'_{\alpha_\ell}\right]\;.
 \eea
 \end{subequations}
The reason is that whenever two non-consecutive $E'$ appear, the operator vanishes due to (\ref{ESE=0}), and the one where no $E'$ appears vanishes as well due to (\ref{Sprime=0}).

Given the fact that $E'' S'' E''=0$ and that $E'' E''=E''$, we immediately obtain (\ref{simple1}) for ${\cal U}_{JJ'}$ with $a=\Phi$ and $b=\rho$. Now, let us consider the lhs of (\ref{simple2}). When we replace decomposition (\ref{WIsom}) in each of the columns, we will have products of the form $A_1 A_2$, where $A_{1,2}$ will be either $E''$ or one of the $S''_\alpha$. We will show now that $A_1 A_2=A_1 E'' A_2$, which will immediately imply (\ref{simple2}), with $a=\Phi$ and $b=\rho$. This is obvious whenever $A_1$ or $A_2$ is equal to $E''$. Thus, we just have to show that $S''_\alpha S''_\beta=S''_\alpha E'' S''_\beta$. But this follows automatically from (\ref{sprimeforms}), since this product vanishes unless $S''_\alpha$ is of type 2 and $S''_\beta$ is of type 1. In that case, we can write $S''_\alpha=S''_\alpha E''$, from which the statement follows. Note that we have blocked $JJ'< D^4$ times.
\end{proof}

\begin{cor}
\label{cor:simple1}
Equation (\ref{simple1}) holds for any tensor ${\cal U}$ generating MPU
(not necessary simple), where the tensors $a,b$ are the left and right fixed
points, $\Phi$ and $\rho$,  of the transfer operator $E$.  \end{cor}
\begin{proof}
Take $J\le D^2$ as in the previous Proposition and consider $U^{(J+1)}$. The LHS of (\ref{simple1}) is simply
 \be
 \frac{1}{d^J}\tr_{1,\ldots, J}(U^{(J+1)\,\dagger}U^{(J+1)}) ,
 \ee
which is clearly the identity, since $U^{(N)}$ is unitary.
\end{proof}

This in turn implies:
\begin{cor}
Let ${\cal U}$ be a tensor generating MPU. If its blocked tensor ${\cal U}_k$ is simple,
the same is true for all ${\cal U}_{k'}$, $k'\ge k$.
\end{cor}

% ============================================================
\subsection{\label{sec:standard-form}
Standard form}

In the previous section we have seen that, starting with a tensor
generating MPU in CF, it is enough to block $k< D^4$ sites to have a simple
tensor. In this subsection we will characterize the minimal such $k$ in
terms of the ranks of two matrices associated to ${\cal U}$. As a
corollary we will obtain an (essentially unique) standard form for simple
tensors. Another corollary will show that any MPU is a quantum cellular
automaton (QCA) and vice versa.

We consider two different singular value decompositions of a (not
necessarily simple) tensor ${\cal U}$ generating MPU, corresponding to two
different choices of how we combine the indices of ${\cal U}$ in order to
write it as matrices, ${\cal M}_1$ and ${\cal M}_2$, respectively. In the
first decomposition, we combine the left auxiliary and down physical
index, as well as the other two, whereas in the second, the other way
around. Graphically, the separation in two sets of indices is represented
by a dashed line:
\be
 \label{II_SVD}
 {\includegraphics[height=3.0em]{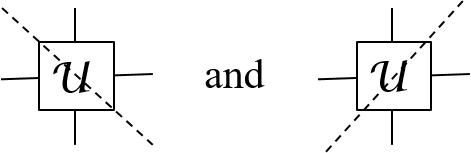}}
\ee

\begin{defn}
\label{defnrl}
We denote by $r$ and $\ell$ the rank of ${\cal M}_{1,2}$, respectively.
\end{defn}

Using a singular value decomposition, we can write 
\begin{equation}
\label{eq:sf-svd}
{\cal M}_{1,2}=V^\dagger_{1,2} D_{1,2} U_{1,2}\ , 
\end{equation}
where the $U$'s and $V$'s are
isometries, fulfilling $V_iV_i^\dagger=U_iU_i^\dagger = \Id$, and
$D_{1,2}$ are diagonal positive matrices, of dimensions $r$ and $\ell$,
respectively. We define now matrices $X_i$ and $Y_i$ such that ${\cal M}_i
= X_i Y_i$ in terms of those decompositions. We choose $X_2=V_2^\dagger$
and $X_1= V_1^\dagger \left[V_1 (\Id\otimes \rho)
V_1^\dagger\right]^{-1/2}$,  so that we have
\be
 \label{Y1Y1X1X1}
 X_2^\dagger X_2 = \Id, \quad X_1^\dagger (\Id\otimes \rho) X_1 = \Id.
 \ee
Recall that $\rho$ is a diagonal matrix with coefficients $\rho_n$ that appear in the right eigenvector of the transfer matrix (\ref{eq:transfer-op}), see (\ref{II_UTransfer}). Since $\rho_n>0$, the square roots and inverses are well defined. Additionally, defining
 \begin{subequations}
 \label{Z1Z2}
 \bea
 Z_2 &=& U_2^\dagger D_2^{-1}\\
 Z_1 &=& U_1^\dagger D_1^{-1} \left[V_1 (\Id\otimes \rho) V_1^\dagger\right]^{-1/2}
 \eea
 \end{subequations}
we have
\be
 \label{YZ=1}
 Y_1 Z_1 = Y_2 Z_2 = \Id
 \ee
and therefore the $Y_i$ are invertible.

Graphically, we can identify some tensors with those matrices
 \begin{subequations}
 \label{XY}
 \bea
 \label{X1Y1}
  && {\includegraphics[height=3.2em]{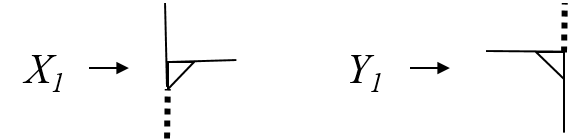}}\\
  \label{X2Y2}
  && {\includegraphics[height=3.2em]{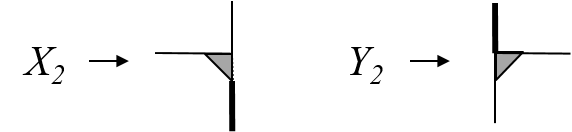}}
 \eea
 \end{subequations}
Note that we have represented all those tensors by triangles with the hypotenuse below the triangle. We have also filled the triangles corresponding to the second decomposition, and use thick dotted and solid lines to represent the bonds with dimensions $r$ and $\ell$, respectively. With this notation, we can re-express (\ref{Y1Y1X1X1}) as
 \be
 \label{X1X2b}
 {\includegraphics[height=3.8em]{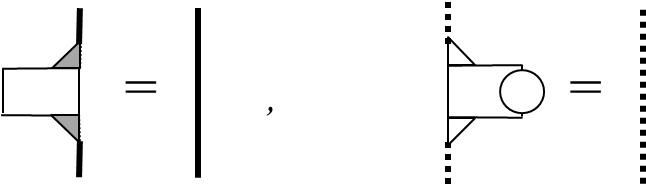}}
 \ee

We can thus decompose the tensor ${\cal U}$ as
\be
 \label{SVDforms2}
 {\includegraphics[height=4.5em]{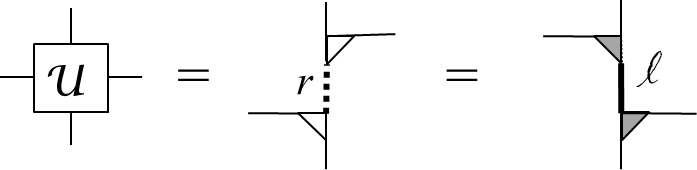}}
\ee

In terms of the previous decompositions, we define
\begin{subequations}
\label{uuvv}
 \bea
 \label{uu}
 u&=&\raisebox{-14pt}{\includegraphics[height=2.8em]{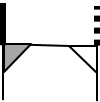}},\\
 \label{vdagger}
 v&=&\raisebox{-12pt}{\includegraphics[height=2.8em]{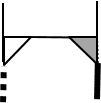}}
 \eea
\end{subequations}
Note that both $u$ and $v^\dagger$ map a space of $d^2$ dimension onto
another one of dimension $r\ell$.

\begin{lem}
\label{lemuisometry}
For any tensor ${\cal U}$, $u^\dagger u=\Id$ and hence $r\ell \ge d^2$
\end{lem}

\begin{proof}
Using (\ref{X1X2b}) we have
\be
 \label{uUnitary}
 {\includegraphics[height=6.0em]{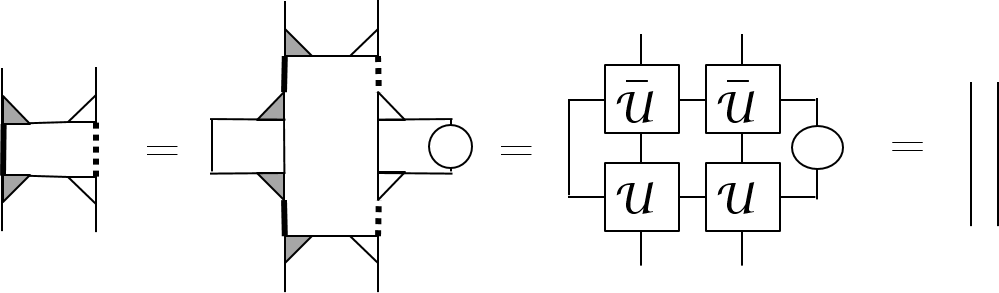}}
\ee
so that $u^\dagger u=\Id$. For the last equality we have used that choosing $k\le D^4$ such that ${\cal U}_k$ is simple (see Proposition \ref{blockingsimple}), the lhs is equal to $1/d^k{\rm tr}_{1,\ldots,k}[U^{(k+2)\dagger} U^{(k+2)}]=\Id$.
\end{proof}

Note that this lemma implies that $u$ is an isometry. In order for it to
be unitary, it is required that $r\ell=d^2$. Now, we are in the position
of stating the main result of this section:

\begin{thm}
\label{ThmFund1}
The following are equivalent for a tensor ${\cal U}$ generating MPU  with the above definitions.
\begin{enumerate}
\item ${\cal U}$ is simple.
\item $r \ell =d^2$.
\item $u$ is unitary.
\item $v$ is unitary.
\end{enumerate}
\end{thm}

\begin{proof}

\noindent
$1\to 2$: Let us assume that ${\cal U}$ is simple. Using (\ref{simple}) we have
\be
 \label{vUnitary}
 {\includegraphics[height=6.2em]{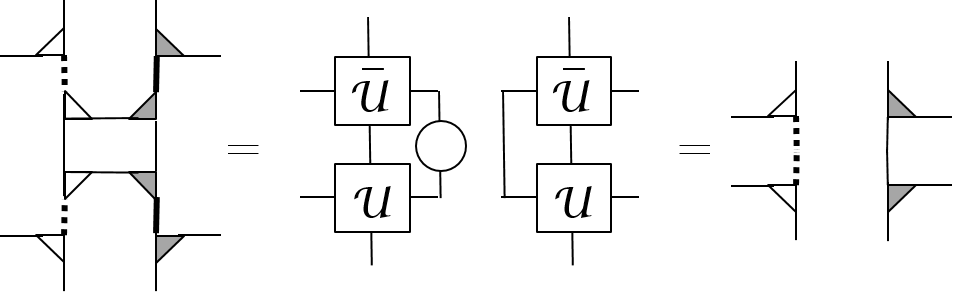}}
\ee
where in the last step we have used (\ref{X1X2b}). Employing (\ref{XY}) we have
\be
(Y_1\otimes Y_2)^\dagger v^\dagger v(Y_1\otimes Y_2)= (Y_1\otimes Y_2)^\dagger (Y_1\otimes Y_2)
\ee
Conjugating both sides of this equation with $(Z_1\otimes Z_2)^\dagger$
and $(Z_1\otimes Z_2)$ and using (\ref{YZ=1}) we obtain that $v^\dagger
v=\Id$. This implies that $r\ell \le d^2$, which together with Lemma
\ref{lemuisometry} implies $r\ell=d^2$.

\noindent
$2\to 3$: If we assume now $r\ell = d^2$, through Lemma \ref{lemuisometry} we immediately get that $u$ is unitary (note that $u^\dagger u=\Id$).

\noindent
$3\to 4$: If $u$ is unitary, then since $U^{(2)}$, which is also unitary,
is just the product of $u$ with $v$ translated by one site, we deduce that $v$
is also unitary.

\noindent
$4\to 1$: Finally, if $v$ is unitary, by Corollary \ref{cor:simple1} it is enough to show (\ref{simple2}),  which is just the first equality in (\ref{vUnitary}). The fact that $v$ is unitary provides us with the equality between the LHS and the RHS of (\ref{vUnitary}). By inserting (\ref{X1X2b}) in the RHS [as done in (\ref{UisUnitary})] we obtain the desired equality with the middle term in (\ref{vUnitary}).
\end{proof}

Equipped with this theorem, we have a practical way of computing $k_0$,
namely the smallest $k$ such that ${\cal U}_{k}$ is simple. We just have
to block and compute $r_k$ and $\ell_k$ for $k=1,2,\ldots$ until we obtain
$r_{k_0}\ell_{k_0} =d_{k_0}^2$; as we have shown in the previous subsection,
$k_0\le D^4$.  Additionally, it also gives us a practical way to check if a
tensor ${\cal U}$ generates MPUs: after blocking and reaching
$r_{k_0}\ell_{k_0}=d_{k_0}^2$, we just have to check if the resulting
tensor is simple.  This theorem also shows that by blocking a simple
tensor ${\cal
U}$ two times we have 
\be
 \label{StandardForm}
 {\includegraphics[height=4.4em]{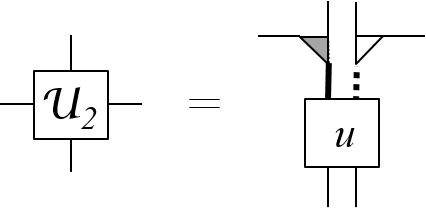}}
\ee
where $u$ is a unitary and $v$, as defined in (\ref{vdagger}), too. This will define a standard form for all MPU (after blocking 2$k_0$ times).

\begin{defn}
\label{SF}
Given a simple tensor, ${\cal U}$, we define the {\em standard form} (SF) of ${\cal U}_2$ as (\ref{StandardForm}).
\end{defn}

Now, we derive the fundamental theorem of MPU, namely the relation between two tensors generating MPU in SF. For that, let us consider two simple tensors, ${\cal U}$ and ${\cal V}$, generating the same MPU, and with standard forms
 \be
 \label{Twotensorsstandard}
 {\includegraphics[height=4.4em]{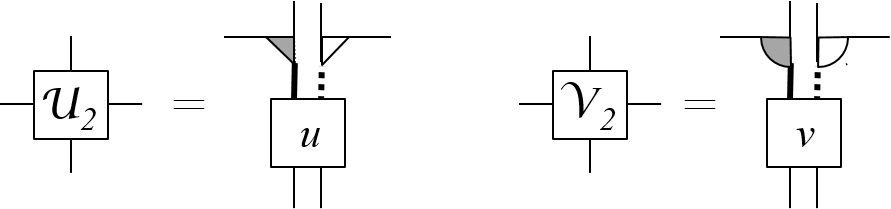}}
\ee

Using the fact that $X_1$ is invertible, it is straightforward to prove
the following result, which we call the fundamental theorem of MPU:

\begin{thm}[Fundamental Theorem of MPU]
\label{FundamentalMPU}
Given two simple tensors, ${\cal U}$ and ${\cal V}$ with standard forms
(\ref{Twotensorsstandard}), they generate the same MPU, i.e.\
$U^{(N)}=V^{(N)}$ for all $N$, iff there exist
unitaries $x,y$ and $z$, such that \begin{subequations}
 \bea
 \label{SFuu}
 &{\includegraphics[height=5.6em]{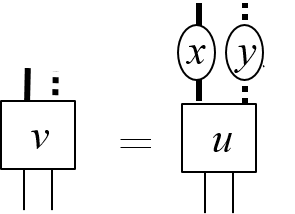}}&,\\
 \label{SFvv}
 &{\includegraphics[height=4.6em]{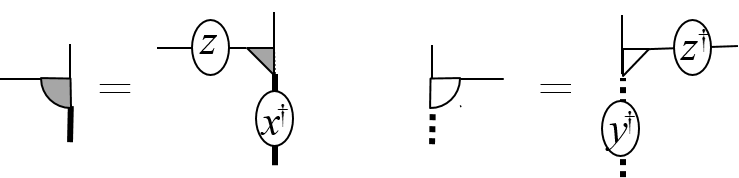}}&
 \eea
\end{subequations}
\end{thm}
Note that the unitarity of $z$ follows from the fact that both tensors are
in CFII \cite{cirac:mpdo-rgfp}.

Using this standard form, we can also identify the tensors $X$ and $Y$ corresponding to its singular value decompositions
 \be
 \label{StandarForm3}
 {\includegraphics[height=6.0em]{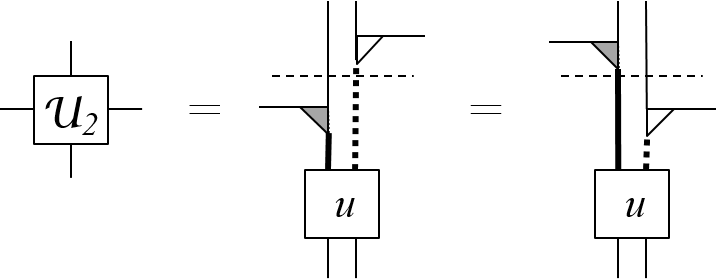}}
\ee
where for the sake of clarity, we have separated with a thin line the two pieces of the singular value decomposition. 

We finish this section making the announced identification of MPU's and QCA. We refer to the Appendix for the formal definition of QCA and the proof of the following

\begin{cor}
\label{QCA}
Any MPU (with finite bond dimension) is a 1D QCA and viceversa.
\end{cor}

% ============================================================
\section{Index}

Corollary \ref{QCA} allows us to apply all known results on 1D QCA to
MPU's. In particular, the index theorem~\cite{gross:index-theorem}, which associates a
rational number to each QCA with the property that two automata have the
same index iff there is a continuous path of QCA connecting them. That is,
the index characterizes the  different possible {\it phases} of QCA.

In this section, we will show how the index theorem can be easily obtained
in the framework of MPU, by using the results obtained in the previous
section. Our approach has the additional benefit of providing a definition
of the index which is very easy to compute. In the Appendix we will show
how our definition coincides with the one given
in Ref.~\cite{gross:index-theorem} for QCA.

Let start defining the index of a tensor ${\cal U}$ in CF generating MPUs:
\begin{defn}
\label{def:index}
Take any $k$  so that the blocked tensor ${\cal U}_{k}$ is simple. Define $r$ and $\ell$ for ${\cal U}_{k}$ as in Definition \ref{defnrl}.
The index of ${\cal U}$ is defined as ${\rm ind}=\frac{1}{2}(log_2(r)-log_2(\ell))$.
\end{defn}

As compared to \cite{gross:index-theorem}, we have included the logarithm in this definition since then the index has a more information-theoretical interpretation: It is the net flow of information to the right. Note that since $r\ell=d^2$, ${\rm ind}=\log_2(r/d)=-\log_2(\ell/d)$.

The first thing to notice is that

\begin{prop}
\label{index-well-defined}
The index is well defined, that is, it does not depend on $k$.
\end{prop}
\begin{proof}
We denote by $k_0$ the smallest $k$ so that ${\cal U}_k$ is simple and take $k>k_0$. Clearly, ${\cal U}_{k}$ is the result of  blocking ${\cal U}_{k_0}$ with ${\cal U}_{k-k_0}$.
In order to determine the index for such particular $k$, we combine indices to build a matrix and compute the rank as explained in the previous section. Graphically, for the first decomposition we have
\be
 \label{}
 {\includegraphics[height=8.5em]{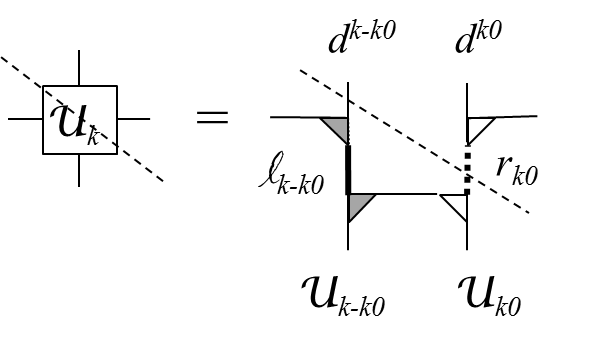}}
\ee
and similarly for the other one. According to this figure, the right rank, $r_k$ is upper bounded by $d^{k-k_0}r_{k_0}$. If we build the matrix the other way round, $\ell_k$ is upper bounded by $d^{k-k_0}\ell_{k_0}$. Applying Theorem \ref{ThmFund1}, since ${\cal U}_k$ is simple, we have that $r_k\ell_k=d^k$, and thus $r_k=d^{k-k_0}r_{k_0}$ and $\ell_k=d^{k-k_0}\ell_{k_0}$, so that the index is $\frac{1}{2}(\log(r_k)-\log(l_k))=\frac{1}{2}(\log(r_{k_0})-\log(\ell_{k_0}))$.
\end{proof}

Before stating and proving the Index Theorem for MPU, we need the following definition of equivalent tensors. Let us consider two tensors, ${\cal U}$ and ${\cal V}$, of physical dimensions $d_{a,b}$, respectively, generating MPU's, and let us denote by $p_{a,b}$ two coprimes such that $p_a d_a=p_b d_b$. We also denote by $\Id_x$ the identity operator acting on a Hilbert space of dimension $x$, and by ${\cal U}^{(x)}={\cal U}\otimes \Id_x$, i.e. the tensor generating the MPV $U^{(N)}\otimes \Id_x^{\otimes N}$.

\begin{defn}
\label{def:strictly-equivalent-tensors}
Two tensors ${\cal U}$ and ${\cal V}$ in CF are {\em strictly equivalent} if
$d_a=d_b$ and there exist a continuous path ${\cal W}(p)$ of tensors, not
necessarily in CF, with $p\in[0,1]$ such that ${\cal W}(0)={\cal U}$ and
${\cal W}(1)={\cal V}$.
\end{defn}

\begin{defn}
\label{def:equivalent-tensors}
Tensors ${\cal U}$ and ${\cal V}$ are {\em equivalent} if there exists
some $k\in \mathds{N}$ and $p_a$, $p_b$ such that ${\cal U}^{(p_a)}_{k}$
and ${\cal V}^{(p_b)}_{k}$ are strictly equivalent.
\end{defn}
This last definition means that by attaching an ancilla of dimension
$p_{a,b}$ and blocking, respectively, they can be locally transformed into
each other.

Rather than saying that two MPUs $\mathcal U$ and $\mathcal V$ are
equivalent or strictly equivalent, we will sometimes also say that they
are in the same phase (or class) with regard to equivalence or strict
equivalence; we will omit the notion of equivalence and just say that
$\mathcal U$ and $\mathcal V$ are in the same phase when it is clear from
the context.

It is important to make clear that we do not impose that the tensors in
the path are in CF to allow the most general interpolation path, along
which e.g.\ the bond dimension in CF might change, such as for the
examples in Eq.~(\ref{SFu1u3}).  Since, however, the index is defined
using the corresponding CF tensor, we will need the following result.

\begin{prop}\label{prop:continuity-index}
Given a continuous map of tensors $[0,1]\ni x\mapsto {\cal W}(x)$, not
necessarily in CF, generating MPU, consider the CF and the ranks
$r(x)$, $\ell(x)$ of the associated simple tensors. Then both $r(x)$,
$\ell(x)$ are constant.  
\end{prop}

\begin{proof}
Let us consider any tensor ${\cal W}$, not necessarily in CF, generating
MPU. By blocking  $D^4$ sites we can assume that its CF is simple and that
its transfer operator is a rank-one map with positive semidefinite left
and right fixed points which we call $L$ and $R$, respectively.  That is, as
a completely positive map, the transfer operator is $X\mapsto
\mathrm{tr}(LX)R$, where $\mathrm{tr}(LR)=1$. This implies that 
\[
\includegraphics[width=3cm]{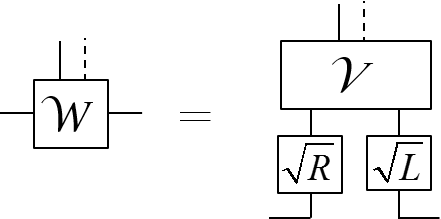}
\]
for an isometry $V$ and hence, if we call $P$ and $Q$ the projectors onto the support of $L$ and $R$ respectively, we have that ${\cal W}P={\cal W}=Q{\cal W}$.

The first thing we will show is how we can use $P$ and $Q$ to obtain the
the CF of $\mathcal W$.  For that we invoke Jordan's Lemma, which ensures
a decomposition of the space $\mathbb{C}^D=(\bigoplus_i \mathbb{C}^2)\oplus
\mathbb{C}^k$ such that in that basis, $P=\bigoplus_i |0\rangle\langle
0|_i\oplus R$,  $P=\bigoplus_i |v_i\rangle\langle v_i|\oplus S$, where $R$
and $S$ are commuting projectors on $\mathbb{C}^k$.
Let us define the projector $\tilde{P}:=\bigoplus_i |0\rangle\langle 0|_i\oplus RS$. 
We have the following properties:
\begin{enumerate}[(i)]
\item $P\tilde{P}=\tilde{P}$.
\item $PQ=\tilde{P}Q$.
\item $QP=Q\tilde{P}$.
\item There exists an invertible $Y$ so that $\tilde{P}Q Y=\tilde{P}$.
\end{enumerate}

We claim that $\tilde{\cal W}:=\tilde{P}{\cal W}\tilde{P}$
is the CF of ${\cal W}$, when restricted to the range of $\tilde P$.
To see this, it is enough to show that both $\tilde{\cal W}$
and ${\cal W}$ define the same MPU for all $N$ and that the left and right
fixed points of the transfer operator of $\tilde{\cal W}$ are full rank.
The latter is obvious since the transfer operator of $\tilde{\cal W}$ is
trivially $X\mapsto {\rm tr}(\tilde{P}L\tilde{P}X)\tilde{P}R\tilde{P}$,
which implies that its left and right fixed points are, respectively,
$\tilde{P}L\tilde{P}$ and $\tilde{P}R\tilde{P}$, which are both full rank
in the range of $\tilde{P}$. The former is obvious from \[
\raisebox{-0.3cm}{\includegraphics[width=6.5cm]{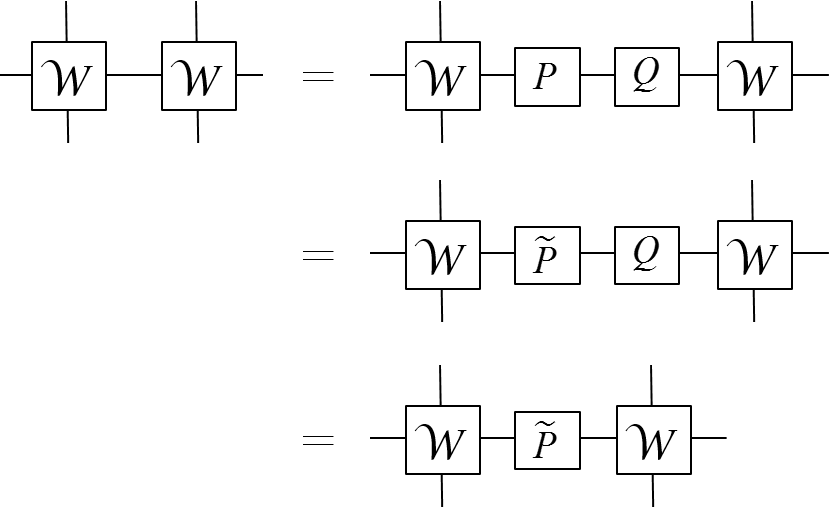}} .  \]

The next step is to define the tensor $\hat {\cal W}=L{\cal W}R$. We will
prove that the rank along the cut 
\[
\includegraphics[width=1.5cm]{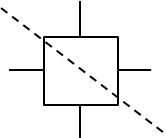}
\]
for $\hat {\cal W}$ and $\tilde{\cal W}$ is exactly the same (and
similarly for the other cut). To this end, we take invertible matrices
$X,Y,Z$ so that $XL=P$, $RZ=Q$ and $PQ Y=\tilde{P}$ (see (ii) and (iv)
above).  Then 
\begin{equation*}
X\hat{\cal W}ZY= P{\cal W}QY=PQ{\cal W}PQY=\tilde{P}{\cal W}\tilde{P}=\tilde{\cal W}\; ,
\end{equation*}
which proves the claim. Moreover, since the CFII of $\tilde{\mathcal W}$
(restricted to the range of $\tilde P$) is constructed as 
$(\tilde PL\tilde P)^{1/2}\tilde W (\tilde PL\tilde P)^{-1/2}$, and
$\tilde PL\tilde P$ is full rank of the range of $\tilde P$, the same
holds for the corresponding ranks in CFII used to define the index, i.e.,
$r(x)$ and $\ell(x)$ are given by the corresponding ranks of
$\hat{\mathcal W}(x)$.

Finally, given the continuous path ${\cal W}(x)$, it is clear that
$\hat{\cal W}(x)$ is also continuous in $[0,1]$ since
\[
\raisebox{-0.3cm}{\includegraphics[width=\columnwidth]{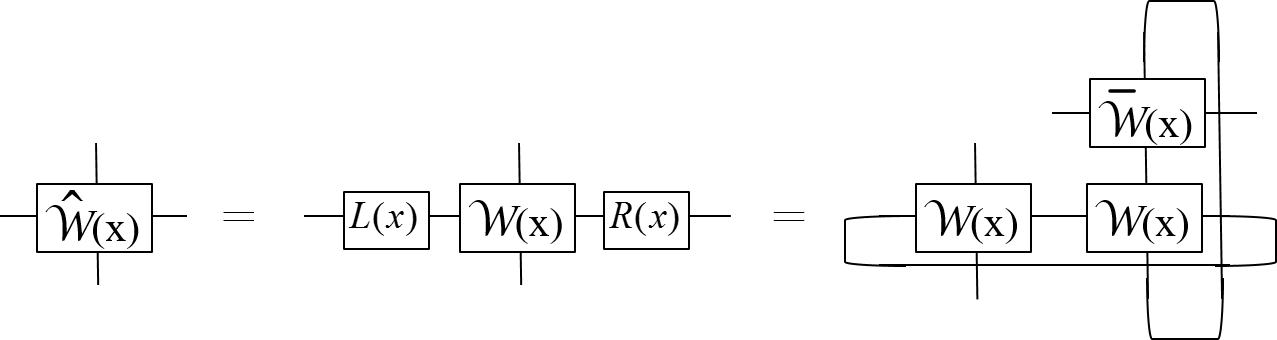}}
.
\]
Since the rank of a continuous path of matrices is lower semicontinuous,
both $r(x)$ and $\ell(x)$ are lower semicontinuous functions on $[0,1]$.
Moreover, $r(x)=\frac{d^2}{\ell(x)}$, which implies that both $r(x)$ and
$\ell(x)$ are also upper semicontinuous and  hence continuous. Since they
take integer values, they must be constant.
\end{proof}

We can now state and prove the main result of this Section.

\begin{thm}[Index Theorem]\quad
\label{IndexTh}
\begin{enumerate}[(i)]
\item The index does not change by blocking.
\item The index is additive by tensoring and composition.
\item The index is robust, i.e. by changing continuously ${\cal U}$ one cannot change it.
\item Two tensors have the same index iff they are equivalent.
\end{enumerate}
\end{thm}

\begin{proof}
(i) has been proven in Proposition \ref{index-well-defined}.
\\[0.5ex]
(ii) The case of tensoring is trivial. To prove it for the concatenation of two MPU, we consider ${\cal U}$ and ${\cal U'}$, with indices ${\rm ind}_{\cal U}$ and ${\rm ind}_{\cal U'}$, respectively. We define
 \be
 {\includegraphics[height=4.8em]{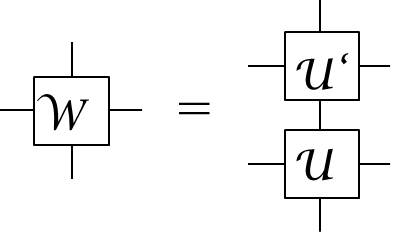}}
\ee
where on the left we have double auxiliary indices, and denote by ind$_{\cal W}$ its index. For that, we block $k$ times until ${\cal U}, {\cal U'}$ and ${\cal W}_k$ are simple. We denote by $r_k,\ell_k$ and $r_k',\ell_k'$ the corresponding ranks of ${\cal U}_k$ and ${\cal U}_k'$. We now express ${\cal U}_{2k}$ and ${\cal U}'_{2k}$ in SF and use them to build ${\cal W}_{4k}$. We consider the matrix by joining indices in order to compute its right and left ranks, $r''$ and $\ell''$,
 \be
 {\includegraphics[height=9.0em]{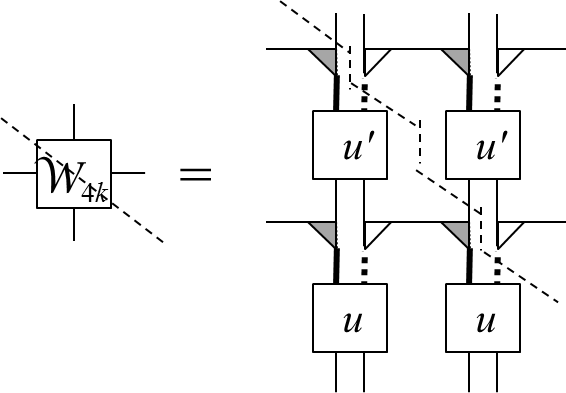}}
\ee
We have $r''\le d_k^2 r_k r_k'$ and, similarly, $\ell''\le d_k^2 \ell_k \ell_k'$. Since ${\cal W}_{4k}$ is simple, $r''\ell''=d_k^8$ we conclude that ${\rm ind}_{\cal W}=\frac{1}{2}(\log(r'')-\log(\ell''))={\rm ind}_{\cal U}+{\rm ind}_{\cal U'}$.
\\[0.5ex]
(iii) has been proven in Proposition \ref{prop:continuity-index}.
\\[0.5ex]
(iv) Let us consider two tensors ${\cal U}$ and ${\cal V}$ and by ${\rm
ind}_{\cal U}$ and ${\rm ind}_{\cal V}$ their indices.

(only if): Let us assume that ${\rm ind}_{\cal U}\ne {\rm ind}_{\cal V}$.
Using (i) and (ii), and the fact that the identity operator has index equal
to zero we have that the indices of ${\cal U}^{(p_a)}_k $  and ${\cal
V}^{(p_b)}_k$ coincide with those of ${\cal U}$ and ${\cal V}$, and thus
are different for any $k$. Thus, according to (iii), they cannot be
continuously deformed into each other, so that they are not equivalent.

(if): if ${\rm ind}_{\cal U}= {\rm ind}_{\cal V}$, then we can find some
$k$ such that both ${\cal U}^{(p_a)}_k$ and ${\cal V}^{(p_b)}_k$ are
simple and have the same physical dimension. Then, the corresponding $u$
and $v$ defined in (\ref{uuvv}) have the same input and output dimensions,
so that they can be smoothly connected to each other, and thus so can
${\cal U}^{(p_a)}_k$ and ${\cal V}^{(p_b)}_k$. Hence they are equivalent.
\end{proof}

Let us now show that any continuous path corresponds to a continuous path
in SF.
\begin{cor}
\label{coruvSF}
For every continuous path of tensors $\mathcal U(p)$ -- not necessarily in
CF -- there exists $k_0\le D^4$ and a continuous path $u(p)$ and $v(p)$
such that the MPU $U^{(2k_0N)}(p)$ generated by $\mathcal U(p)$ has a SF
given by $u(p)$ and $v(p)$. 
\end{cor}

\begin{proof}
Let $\hat{\mathcal U}(p)$ be the CFII of $\mathcal U(p)$. Choose $k_0\le
D^4$ such that for all $p$, $\hat{\mathcal U}_{2k_0}(p)$ is in SF, with
unitaries $u(p)$ and $v(p)$. Note that the Index Theorem implies that
$u(p):\mathbb C^d\otimes \mathbb C^d\to \mathbb C^\ell\otimes \mathbb C^r$ 
and 
$v(p):\mathbb C^r\otimes \mathbb C^\ell\to \mathbb C^d\otimes \mathbb C^d$ 
map between tensor product spaces of fixed dimensions,
independent of $p$. Consider $U(p):=U^{(2k_0)}(p)$, the
MPU generated by $\mathcal U(p)$ [and equivalently $\hat{\mathcal U}(p)$]
on $2k_0$ sites. Since $\mathcal U(p)$ is continuous, so is $U(p)$. Thus,
for any $x\in[0,1]$ and $\varepsilon>0$ there exists a $\delta>0$ such that for all
$|x-y|<\delta$, 
\[
\|U(x)U(y)^\dagger - \openone\|_1\le \varepsilon\ .
\]
Substituting $U(p)=v_{23}(p)v_{41}(p)u_{12}(p)u_{34}(p)$, where the
subscripts denote the $k_0$-blocked sites, and using unitary invariance of
the trace norm, this yields
\begin{equation}
\label{eq:uu-vv-cont}
\begin{aligned}
\|u_{12}^{\phantom\dagger}(x)u_{12}^\dagger(y)\otimes 
u_{34}^{\phantom\dagger}(x)u_{34}^\dagger(y) -\hspace*{3cm} \\
v_{23}^{\dagger}(x)v_{23}^{\phantom\dagger}(y)\otimes 
v_{41}^{\dagger}(x)v_{41}^{\phantom\dagger}(y)\|_1
\le \varepsilon\ .
\end{aligned}
\end{equation}
Now pick $\ket{\alpha}$, $\ket{\beta}$, $\ket{\gamma}$, $\ket{\delta}$
s.th.\ 
$C:=\bra{\alpha,\beta}_{34}
u_{34}^{\phantom\dagger}(x)\times u_{34}^\dagger(y)\ket{\gamma,\delta}_{34}$
satisfies $|C|\ge 1/\sqrt{d}$, where $d$ is the dimension of the blocked
site -- e.g., by choosing  $\ket{\alpha,\beta}$ the leading Schmidt vector of 
$u_{34}^{\phantom\dagger}(x)u_{34}^\dagger(y)\ket{\gamma,\delta}_{34}$.
Since $\|VXW^\dagger\|_1\le \|X\|_1$ for isometries $V$ and $W$, we obtain
with $V=\openone_{12}\otimes \bra{\alpha,\beta}_{34}$ and
$W^\dagger=\openone_{12}\otimes\ket{\gamma,\delta}_{34}$ from
Eq.~(\ref{eq:uu-vv-cont}) $\|C\,
u_{12}^{\phantom\dagger}(x)u_{12}^\dagger(y) -\tilde r_{1}(x,y)\otimes
\tilde s_{2}(x,y)\|_1\le\varepsilon$ [with $\tilde r_1(x,y)=\bra\beta_4
v_{14}^\dagger(x)v_{14}(y)\ket{\delta}_4$, and similarly for $\tilde s_2(x,y)$],
and thus
\[
\|u_{12}(x) - 
(r_{1}(x,y)\otimes s_{2}(x,y))\, u_{12}(y) \|_1\le
\sqrt{d}\, \varepsilon\ .
\]
We choose the relative norms of $r_1(x,y)$ and $s_2(x,y)$ such
that $\|r_1(x,y)\|_\infty=1$.

While $r_1(x,y)$ and $s_2(x,y)$ are generally not unitary, we can replace
them by the unitaries $r_1'(x,y)$ and $s_2'(x,y)$ obtained by setting their
singular values to $1$, with a continuity bound 
\[
\|u_{12}(x) - 
(r_{1}(x,y)\otimes s_{2}(x,y))\, u_{12}(y) \|_1\le
(4d^2+1)\sqrt{d}\, \varepsilon\ .
\]
This is easily shown from $\|u_{12}(x)u_{12}(y)^\dagger-r_1\otimes
s_2\|_1\le\sqrt{d}\varepsilon$, which (using $\|r_1\|_\infty=1$) can be
used to bound their singular values, yielding $\|r_1\otimes
s_1-r_1'\otimes s_1'\|_1\le 4d^2\sqrt{d}\varepsilon$. 

We thus see that $u_{12}(x)$, modulo local unitaries, is continuous in
$x$.  Using a classical result of Montgomery and Yang
\cite{montgomery:existence-of-a-slice} (see
Ref.~\cite{brendon:compact-transformation-groups} for a detailed explanation),
it follows that there exists a local unitary gauge $\hat r(x)\otimes \hat
s(x)$ such that $\hat u(x)=(\hat r(x)\otimes\hat s(x))u(x)$ is continuous.  It is
then trivial to see that the corresponding $\hat v(p)$ is continuous as well,
by using the continuity of $U^{(k_0)}(p)= \hat u(p)\mathbb S \hat v(p)\mathbb S$,
where $\mathbb S$ swaps the indices of $\hat v(p)$.  \end{proof}

\section{\label{sec:sym-local-char}
Symmetries: Local characterization}

We will now turn our interest towards symmetries, and study the effect
they have on the classification of MPUs, i.e., what happens when we impose
some symmetry $U^{(N)}=\mathcal S[U^{(N)}]$ on the allowed MPUs. To this
end, we will use the standard form for MPUs (Definition~\ref{SF}) together
with the  local characterization of different MPU representations of the
same unitary -- namely $\mathcal U$ and $\mathcal S[\mathcal U]$ --
established in Theorem~\ref{FundamentalMPU}.

We will consider the following three symmetries:
\begin{enumerate}
\item Conjugation $U^{(N)}=\overline{U^{(N)}}$
\item Time reversal $U^{(N)}={U^{(N)}}^\dagger$
\item Transposition $U^{(N)}={U^{(N)}}^{T}$
\end{enumerate}
The ideas used can however be adapted, for instance to combine the above
symmetries with local symmetry actions, such as in the case of
topological insulators (where additionally the spin is flipped), but also
to systems with other locally characterized symmetries.

In this section, we provide local characterizations of the three
aforementioned symmetries based on Theorem~\ref{FundamentalMPU}. We then
apply these local characterizations in Sec.~\ref{sec:sym-conjugation} to
classify the equivalence classes of MPUs under symmetries.  Finally, we
discuss a range of examples in Sec.~\ref{sec:sym-examples}.

The general procedure to analyze symmetries in MPU is analogous to that
used for MPS~\cite{sanz:mps-syms}. Let us assume that for all $N$, the MPU
generated by a tensor ${\cal U}$ is invariant under some symmetry, that
is, $\tilde U^{(N)}={\cal S}[U^{(N)}]=U^{(N)}$. If we denote by $\tilde
{\cal U}$ the tensor generating $\tilde U^{(N)}$, we have that both ${\cal
U}$ and $\tilde{\cal U}$ generate the same family of MPU. Thus, if we
block until we obtain the SF in both of them, then the resulting tensors
have to be related by virtue of Theorem \ref{FundamentalMPU}. More explicitly, we
will consider a tensor of the form (\ref{StandardForm}), and define $u$
and $v$ as in (\ref{uuvv}).  Thus, $u$ and $v$
fully characterize the MPU, and we just have to find the restrictions
imposed by the symmetry. We will say that $u,v$ define the SF and that
they generate
the MPU.

% ============================================================
\subsection{Conjugation}

We will consider MPU that are invariant under complex conjugation, that
is, $\forall N: U^{(N)}=\overline{U^{(N)}}$. We will consider the SF
defined by $u$ and $v$. This implies that $U^{(N)}$ and
$\overline{U^{(N)}}$ must have the same index, since their corresponding
dimensions $r$ and $\ell$ are equal.

\begin{prop}
\label{prop1conj}
The MPU generated by a tensor in SF are invariant under conjugation iff
there exist unitaries $x,y$ fulfilling either $x=x^T,y=y^T$, or $x=-x^T,y=-y^T$, such that
\begin{subequations}
 \label{conjugation}
 \bea
 \label{uconjugation}
 \bar u&=&(x\otimes y) u\ ,\\
 \label{vconjugation}
 \bar v&=& v (y^\dagger\otimes x^\dagger) \ ,
 \eea
 \end{subequations}
\end{prop}

\begin{proof}
(if) It is immediate.\\[0.5ex]
(only if) Let us assume that $U^{(N)}=\overline{U^{(N)}}\;\forall N$, where the
conjugation is taken in the original basis where the MPU is defined. Using
Theorem \ref{FundamentalMPU}, we obtain (\ref{uconjugation}) and
(\ref{vconjugation}) for some unitaries $x,y$. Using (\ref{uconjugation})
twice, we obtain $\bar x x \otimes \bar y y =\openone$, which is
equivalent to $x=e^{i\phi} x^T$ and $y=e^{-i\phi} y^T$. This implies
$\bar x = e^{-i\phi} x^\dagger$ and thus $\openone = xx^\dagger =
e^{2i\phi} x^T\bar{x}=e^{2i\phi}\openone$, such that $e^{i\phi}=\pm1$.
\end{proof}

Let us now have a closer look at the structure of symmetric and
skew-symmetric unitary
matrices.

\begin{lem}
\label{lemma:conjclass-normalform-continuous}
If a unitary matrix $x$ fulfills $x=x^T$ or $x=-x^T$, then there exist a
symmetric unitary matrix $S$, $S=S^T$, and a $\tilde\Lambda$ 
such that $x=S^T\tilde\Lambda S$. Moreover,
$\tilde\Lambda$ is real and symmetric (for $x=x^T$) or skew-symmetric (for
$x=-x^T$), and $\det(\tilde\Lambda)=1$.
\end{lem}

\begin{proof}
Let us consider first the case $x=x^T$. Since $x$ is unitary, it has a 
spectral decomposition $x= \sum_E e^{-i E} P_E$ with 
$E\in[0,2\pi)$, where the $\{P_E\}_E$ form a complete set of mutually
orthogonal orthogonal projectors, and $P_E=P_E^T$. 
Defining
\[
 S= \sum_E e^{-iE/2} P_E
\]
and $\tilde\Lambda=\openone$, we have indeed that $S$ is unitary,
$S=S^T$, and $x=S^T\tilde\Lambda S$.

Let us now consider $x=-x^T$. Since $x$ is unitary, we can write
\[
 x= \sum_E e^{-iE} (P_E - P_E^T)
\]
where $E\in[0,\pi)$, and the $\{P_E,P_E^T\}_E$ form a complete set of
mutually orthogonal orthogonal projectors. Defining
\[
 S= e^{-i\pi/4}\sum_E e^{-iE/2} (P_E + P_E^T)\ ,
\]
and
\[
 \tilde \Lambda = i \sum_E (P_E - P_E^T)
\]
we have that $S$ is unitary, $S=S^T$, and $x=S^T\tilde\Lambda S$.
Moreover, $\bar{\tilde \Lambda}=-i\sum (\bar P_E-P_E^\dagger) =
-i\sum(P_E^T-P_E)=\tilde\Lambda$, i.e., $\tilde\Lambda$ is real, and
$\tilde\Lambda^T=-\tilde\Lambda$.  Finally, since the eigenvalues of
$\tilde\Lambda$ come in pairs $\pm i$, $\det(\tilde\Lambda)=1$.
\end{proof}

We now show that there always exist a gauge such that $u$ and $v$ can be choosen
with special properties. This will be the basis of the classification of phases
that we will carry out in the next section.

\begin{prop}
\label{corconj}
If the MPU generated by a tensor in SF are invariant under complex
conjugation, then one can choose the SF with unitaries $u'$
and $v'$ such that:
\begin{itemize}
\item Case I: If $x=x^T$ and $y=y^T$ in Prop.~\ref{prop1conj} then
    \be
    \label{ortuv}
    \overline{u'} = u'\,, \quad \overline{v'}=v',
    \ee
    i.e.\ they are orthogonal. The choice of $u'$ and $v'$
    is unique up to local orthogonal transformations; that is, if $u''$ and $v''$ are orthogonal
    and generate the same MPU, then there exist $O_{r,\ell}$ orthogonal, such that
    $u''=(O_\ell\otimes O_r) u'$ and $v''=v' (O_r^T\otimes O_\ell^T)$.
\item Case II: If $x=-x^T$ and $y=-y^T$ in Prop.~\ref{prop1conj}, then
    \be
    \label{sympuv}
    \overline{u'} = (\Sigma_{\ell}\otimes\Sigma_r) u'\ ,\quad
    \overline{v'}= v'(\Sigma_r\otimes\Sigma_\ell)\ ,
    \ee
    where $\Sigma_x=\openone_{x/2}\otimes Y$, $x=r,\ell$,
    and $Y=\left(\begin{smallmatrix}0&1\\-1&0\end{smallmatrix}\right)$.
    This choice is unique up to local symplectic transformations;
    that is, if $u''$ and $v''$ fulfill (\ref{sympuv}) and generate the same MPU, then
    there exist $S_{r,\ell}$  fulfilling $S_x^T \Sigma_x  S_x=\Sigma_x$
    for $x=r,\ell$ such that $u''=(S_\ell\otimes S_r) u'$ and $v''=v' (S_r^T\otimes S_\ell^T)$.
    Note that this case can only occur if $\ell$ and $r$ are even.
\end{itemize}
\end{prop}

\begin{proof}
First, it holds that every unitary with $x=\pm x^T$ can be written as 
$x=U^T \Lambda U$ with $U$ unitary, where for $x=x^T$, $\Lambda=\openone$,
while for $x=-x^T$, $\Lambda=\Sigma$. This follows directly from 
Lemma~\ref{lemma:conjclass-normalform-continuous} and its proof: For
$x=x^T$, this is immediate with $U=S$. For $x=-x^T$, we use that
$\tilde\Lambda$ is real and skew-symmetric and can thus be
block-diagonalized with an orthogonal transformation $O$, $O^T\Sigma
O=\tilde\Lambda$, and then choose $U=OS$. (Alternatively, this can be
derived from the Autonne-Takagi factorization or Youla's decomposition of
(skew-)symmetric matrices under unitary congruence, respectively, using
the unitarity of $x$.)

Now consider $u$ and $v$ describing an MPU satisfying
\eqref{conjugation}, and write $x=U^T\Lambda_\ell U$ and $y=V^T\Lambda_r V$ with
$U$, $V$ unitary. Define $u':=(U\otimes V)u$ and $
v':=v(V^\dagger\otimes U^\dagger)$.  Clearly, $u'$ and $v'$
describe the same MPU, and it is straightforward to check
that $\bar{u'}=(\Lambda_\ell\otimes \Lambda_r){u'}$ and $\bar{v'}=\
v'(\Lambda_r^T\otimes \Lambda_\ell^T)=v'(\Lambda_r\otimes\Lambda_\ell)$,
which proves the first statement for both cases.

Now, if $u'',v''$ generate the same MPU as $u',v'$, according to
Theorem~\ref{FundamentalMPU} they must be related by some unitaries, ie
$u''=(a\otimes b)u'$ and $v''=v'(b^\dagger\otimes a^\dagger)$.  If $u''$
and $v''$ are orthogonal, then so is $a\otimes b$, and thus
$\bar{a}=e^{i\phi}a$, $\bar{b}=e^{-i\phi}b$. Choosing
$O_\ell=e^{-i\phi/2}a$ and $O_r=e^{i\phi/2}b$ finishes the proof for Case
I. If $u''$ and $v''$ fulfill (\ref{sympuv}), then
it is straightforward to show that $(a\otimes b)
(\Sigma_\ell\otimes\Sigma_r)(a\otimes b)=
(\Sigma_\ell\otimes\Sigma_r)$, and thus $a$ and $b$ can be chosen
symplectic by an appropriate phase choice, as
announced in Case II.
\end{proof}

The second case in this proposition may look more intricate that the first one. However,
there exists a simple characterization for the unitaries $u',v'$.

\begin{cor}
\label{cor:sympl-to-orthogonal}
In Case II of Proposition \ref{corconj} we can write
 \begin{subequations}
\label{upvpconj}
 \bea
 u' &=& (\openone_{r\ell/4}\otimes R) {\tilde u}\\
 v' &=& {\tilde v}(\openone_{r\ell/4}\otimes R)
 \eea
 \end{subequations}
with $R=\frac{1-i}{2} [\Id_4 + i(Y\otimes Y)]$ unitary, and ${\tilde u}$ and
${\tilde v}$ orthogonal, where the ordering of the tensor product is such
that the positions of the $Y$ coincide with those in $\Sigma_x$ in (\ref{sympuv}).
\end{cor}

\begin{proof}
This can be obtained by noticing that $R^T(Y\otimes Y)R=\openone_4$.
\end{proof}

We emphasize, however, that the resulting ${\tilde u}$ and ${\tilde v}$ do
not generate the same MPU.  Indeed, we will see that the MPU generated by
$\tilde{u}$, $\tilde{v}$ and by $u'$, $v'$ are in fact in different
phases under conjugation.

% ============================================================
\subsection{Time-reversal}

We consider now the case where $U^{(N)} = U^{(N)\dagger}$. The first thing
to notice is that this is impossible unless the index of the
generating tensor vanishes. Indeed, if it is non-zero, the tensors
generating $U^{(N)}$ and $U^{(N)\dagger}$ cannot be equivalent, since
their indices are opposite, a result that directly follows from Theorem
\ref{IndexTh} and the fact that $U^{(N)}U^{(N)\dagger}=\Id$ has index
equal to zero.

In the following we will consider the SF so that there are two spins per
site, $H_d=H_{d_0}\otimes H_{d_0}$, i.e.\ in site $n$ we have spins
$2n-1$ and spin $2n$.  We will use a subscript to indicate the spin when
required.

\begin{prop}
\label{PropChiral}
The MPU generated by a tensor in SF are invariant under chirality,
$U^{(N)}=U^{(N)\dagger}$ for $N$ even, iff there exists a unitary $x$ such
that
\begin{equation}
\label{UUdaggertrick}
 (\Id_1\otimes v_{23} \otimes \Id_4)(u_{12}\otimes u_{34})=
\pm
 (x_{12}\otimes x_{34}^\dagger) (\Id_1\otimes v_{23}^\dagger \otimes \Id_4)
\end{equation}
\end{prop}
\noindent
We will discuss the relevance of the $\pm1 $ phase for the classification
of phases in Sec.~\ref{sec:necessary-cond-time-rev}.

\begin{proof}
(if) For even $N$, we have
\be
 \label{eq:dagger-sym-cond}
 {\includegraphics[width=\columnwidth]{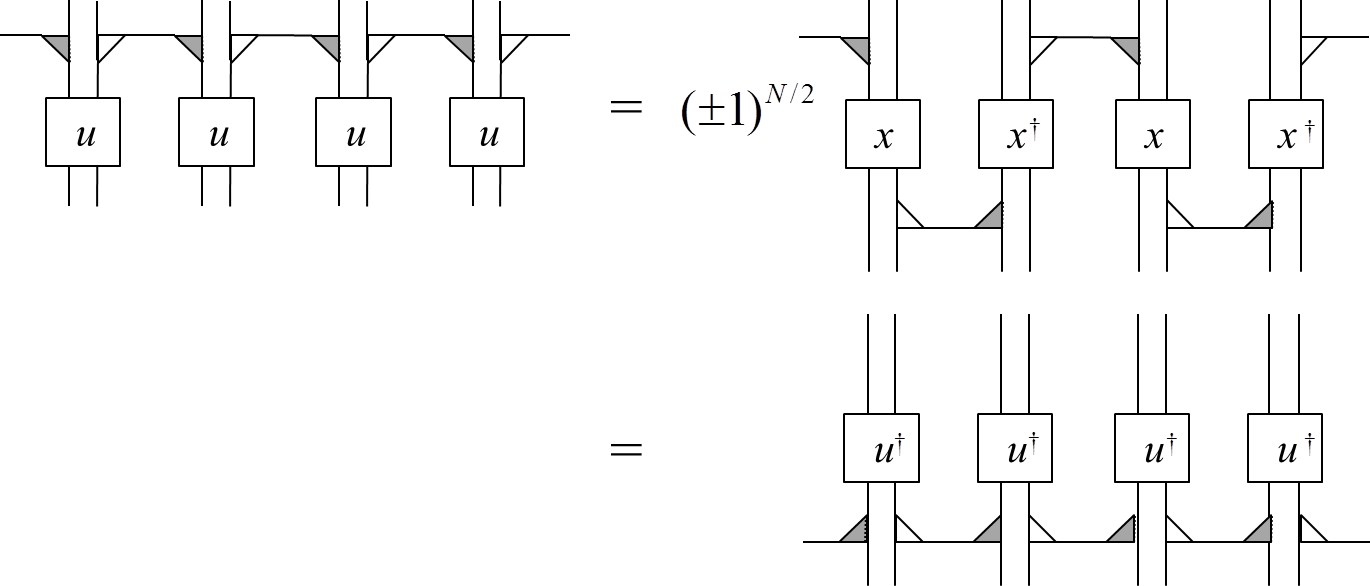}}
\ee
where we first used (\ref{UUdaggertrick}) and then its daggered version.
\\[0.5ex]
(only if) The tensor generating $U^{(N)\dagger}$ is not in SF, so that we cannot use the fundamental theorem. The way around this is to block two tensors, so that
 \be
  \label{eq:block-twosites}
 {\includegraphics[height=4.5em]{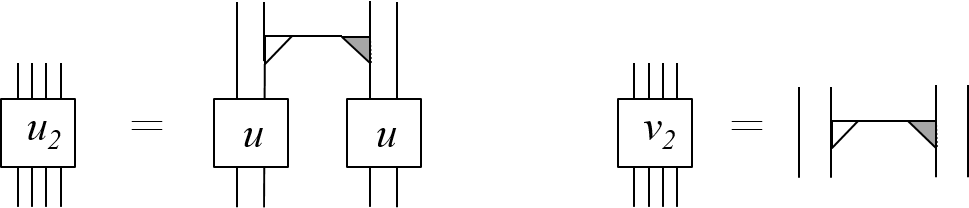}}
\ee
Using Theorem \ref{FundamentalMPU} [Eq. (\ref{SFuu})] for the blocked MPU,
we have $(\Id_1\otimes v_{23} \otimes \Id_4)(u_{12}\otimes u_{34})=
(x_{12}\otimes y_{34}) (\Id_1\otimes v_{23}^\dagger \otimes \Id_4)$.
Analogously, using Eq.~(\ref{SFvv}) we deduce that $(\Id_1\otimes v_{23}
\otimes \Id_4)=(u_{12}^\dagger\otimes u_{34}^\dagger)(\Id_1\otimes
v_{23}^\dagger \otimes \Id_4) (y_{12}^\dagger\otimes x_{34}^\dagger)$.
Substituting the dagger of the second equation into the first one, we
obtain $x_{12}y_{12}\otimes y_{34}x_{34}=\openone$, and thus $y=\pm
x^\dagger$.
\end{proof}

% ============================================================
\subsection{Transposition}

Now we will assume that $U^{(N)}=U^{(N)T}$. Again, we can
only have this symmetry if the index vanishes, since the indices of
$U^{(N)}$ and ${U^{(N)}}^T$ are opposite -- this can be seen by blocking as in
Eq.~(\ref{eq:block-twosites}) and using Definition~\ref{def:index} of the
index in terms of $r$ and $\ell$.

\begin{prop}
\label{prop:local-char-transposition}
The MPU generated by a tensor in SF are invariant under transposition,
$U^{(N)}=U^{(N)T}$ for $N$ even, iff there exists a unitary, $x$, such that 
\be
 \label{UUtransposetrick}
 (\Id_1\otimes v_{23} \otimes \Id_4)(u_{12}\otimes u_{34})=
e^{i\phi}
 (x_{12}\otimes x_{34}^T) (\Id_1\otimes v_{23}^T \otimes \Id_4)
 \ee
\end{prop}

\begin{proof} The proof is completely analogous to that of Proposition
\ref{PropChiral}, just that the phase cannot be fixed to $\pm 1$. This can
also be understood from the fact that in the derivation of the ``if''
direction, the phase cancels out, unlike in the proof of
Prop.~\ref{PropChiral} where it adds up.
\end{proof}

Note that the phase $e^{i\phi}$ in Eq.~(\ref{UUtransposetrick}) bears no
relevance for the classification of phases: It can be continuously changed
to $\phi=0$ while keeping the symmetry under transposition by replacing
$u\leadsto e^{-i\theta/2}u$ and changing $\theta$ from $0$ to $\phi$.

% ============================================================
\subsection{\label{othersymmetries}
Other symmetries}

We conclude this section by noting that other important symmetries
involving any local unitary operator, $Q=Q^\dagger=Q^T$, can be reduced to
the ones above. For instance:
\begin{subequations}
 \label{SSymmetries}
 \bea
 \label{SSymmetries-dagger}
 \left[Q^{\otimes N}\right] U^{(N)\dagger} \left[Q^{\otimes N}\right]&=& U^{(N)}\\
 \label{SSymmetries-transpose}
  \left[Q^{\otimes N}\right] U^{(N)T} \left[Q^{\otimes N}\right]&=& U^{(N)}
 \eea
 \end{subequations}
can be reduced to the previous ones by defining $\tilde {\cal U}=  {\cal
U}Q$, so that the MPU generated by those new tensors are invariant under
time reversal and transposition. The most interesting case is the one where
$Q=\mathbb S$, i.e.\ the swap operator
\be
\mathbb S(a \otimes b)\mathbb S=b\otimes a
\ee
exchanging the two physical indices in the
SF (\ref{StandardForm}), as we will see in some examples.

We conclude this section by just mentioning that a similar approach to the
one developed here, based on the fundamental theorem of MPUs can be
applied to other symmetries, like reflection or global symmetries.

% ============================================================

% ==========================================================================
\section{\label{sec:sym-conjugation}Symmetries: Equivalences}

In the following, we consider the equivalence of MPUs under
symmetry constraints. For the case of conjugation, we will classify all possible phases.
In the other cases, we will give necessary conditions for two MPUs to be in the same
phase. In the next section we will illustrate these results with some examples.

In full analogy to the definition of strict equivalence
(Def.~\ref{def:strictly-equivalent-tensors}) and equivalence
(Def.~\ref{def:equivalent-tensors}) of MPU tensors, we will consider the
following two notions of equivalence under symmetry:

\begin{defn}
\label{def:strictly-equivalent-symmetry}
Two tensors ${\cal U}$ and ${\cal V}$ are {\em strictly equivalent under
the symmetry $\mathcal S$} if $d_a=d_b$ and there exist a continuous path
${\cal W}(p)$, $p\in[0,1]$, not necessarily in canonical form, 
such that ${\cal W}(0)={\cal U}$ and
${\cal W}(1)={\cal V}$, where the MPU $W(p)^{(N)}$ defined by $\mathcal
W(p)$ is invariant under $\mathcal S$, $\mathcal
S[W(p)^{(N)}]=S[W(p)^{(N)}]$.
\end{defn}

\begin{defn}
\label{def:equivalent-symmetry}
Two tensors ${\cal U}$ and ${\cal V}$ are {\em equivalent under the
symmetry $\mathcal S$} if there exists some $k\in \mathds{N}$ and $p_a$,
$p_b$ such that ${\cal U}^{(p_a)}_{k}$ and ${\cal V}^{(p_b)}_{k}$ (the
tensors obtained by blocking $k$ sites and adding ancillas of dimension
$p_a$ and $p_b$, respectively) are strictly equivalent under the symmetry
$\mathcal S$.
\end{defn}

As we will see below, unlike the case without symmetries, these definitions give different results. The motivation of keeping both definitions is that the second involve that $N$ is a multiple of $k$, and thus may be restrictive when considering finite systems.
Since we will assume SF for all the MPUs we consider here, we will characterize the equivalence under symmetries in terms of the unitaries $u,v$ that define them:

\begin{lem}
\label{lem:equivalent-symmetry}
Two MPU, $U_{1,2}^{(N)}$ described by simple tensors are strictly
equivalent under the symmetry $\mathcal S$ iff there exist two continuous families
of unitary matrices $\tilde u(p),\tilde v(p)$, $p\in[0,1]$ such that $\tilde u(0)=\tilde u_1$,
$\tilde v(0)=\tilde v_1$, generate $U_{1}^{(N)}$, and  $\tilde u(1)=\tilde u_2$, and $\tilde v(1)=\tilde v_2$ generate $U_{2}^{(N)}$
and the MPU generated by $\tilde u(p),\tilde v(p)$ are invariant under $\mathcal S$.
\end{lem}

\begin{proof}
The proof follows directly from Corollary \ref{coruvSF}.
\end{proof}

%======================================================================
\subsection{Equivalences under Conjugation}

We consider two simple tensors, ${\cal U}_{1,2}$ generating MPUs that are
invariant under conjugation and are interested in determining when they
are equivalent or strictly equivalent under that symmetry. In SF, they are
generated by some $u_1,v_1$ and $u_2,v_2$, which according to
Theorem~\ref{FundamentalMPU} are only fixed up to a unitary gauge
$u_i\leadsto (x\otimes y)u_i$, $v_i=v_i(y^\dagger\otimes x^\dagger)$.
According to Lemma~\ref{lem:equivalent-symmetry}, we thus want to know the
precise conditions under which we can choose $u_1,v_1$ and $u_2,v_2$ such
that they can be connected with a continuous path $u(p),v(p)$ which define
MPUs with the same symmetry.  We will consider two cases: (i) at least one
of $r,\ell$ are odd; (ii) both $r$ and $\ell$ are even.

%======================================================================
\subsubsection{At least one of $r$, $\ell$ are odd}

If at least one of $r$, $\ell$ is odd, only Case I in Proposition
\ref{corconj} can occur. 

\begin{thm}
\label{conjphs}
If either $r$ or $\ell$ is odd,
${\cal U}_{1,2}$ are strictly equivalent under conjugation symmetry iff 
\[
 {\rm det}(u_1v_1)={\rm det}(u_2v_2)\ .
\]
\end{thm}
\noindent
Note that 
$\det((x\otimes y)u_i\, v_i(y^\dagger\otimes x^\dagger)) = 
\det(x\otimes y)\times\det(u_i v_i)\det(y^\dagger\otimes x^\dagger) = 
\det(u_iv_i)$, i.e.\ $\det(u_iv_i)$ is gauge invariant (as required).

\begin{proof}
{(only if)} Starting from Proposition \ref{prop1conj} it immediately
follows that ${\rm det}(u_iv_i)=\pm 1$. Since the sign of this determinant
is a continuous function of $u$ and $v$, we conclude that if they are
equivalent this determinant must be the same. 
\\[0.5ex]
(if) We will explicitly construct the continuous path $u(p)$, $v(p)$. To
this end, we first choose a suitable gauge for the SF 
$u_i$ and $v_i$, $i=1,2$ of the endpoints $\mathcal U_{i}$.
 Since $r$ and $\ell$ are not both even, we are
in Case I of Proposition \ref{corconj}, i.e., we can choose a gauge where
$u_i$ and $v_i$ are orthogonal. Furthermore, if $\det(u_1)=-\det(u_2)$
[and thus
${\rm det}(v_1)=-{\rm det}(v_2)$], and w.l.o.g.\ $r$ odd, we can find an
orthogonal $\ell\times \ell$ matrix, $O_\ell$, with $\det(O_\ell)=-1$,
and change the gauge as $u_2\mapsto (O_\ell\otimes \openone_r) u_2$ and
$v_2\mapsto v_2 (\openone_r\otimes O_\ell^T)$, which multiplies their
determinants by $\det(O_\ell)^r=-1$. 

We thus see that we can always find a gauge with $u_i$, $v_i$ orthogonal
and $\det(u_1)=\det(u_2)$, $\det(v_1)=\det(v_2)$.  Then, since the
group of orthogonal transformations with a fixed determinant is simply
connected, we can always find a smooth orthogonal interpolation $u(p)$,
$v(p)$ between them.  But the fact that it is orthogonal, according to
Proposition \ref{prop1conj}, implies that it satisfies the symmetry, so
that this finishes the proof.  
\end{proof}

\begin{cor}
All ${\cal U}$ are equivalent under conjugation symmetry.
\end{cor}

\begin{proof}
Given $u,v$ generating MPU with conjugation symmetry, if we block once we
will have that the new $u^{(2)},v^{(2)}$ can be written as
$u^{(2)}_{11',22'}= v_{1'2}u_{11'}u_{22'}$ and
$v^{(2)}_{11',22'}=v_{1'2}$, so that ${\rm det}(u^{(2)}v^{(2)})={\rm det}(
uv)^{2 r\ell}=1$. By applying the previous proposition we arrive at the
conclusion.  \end{proof}

%======================================================================
\subsubsection{$r$ and $\ell$ are both even}

Now, both Case I and Case II in Proposition \ref{corconj} can occur. Since
we can continuously change the gauge for $u,v$ without affecting the
symmetry, we will assume that $u,v$ fulfill (\ref{ortuv}) and
(\ref{sympuv}) in Case I and II, respectively. With this choice, 
\be
 {\rm det}(u_i)=\pm1\ ,\quad {\rm det}(v_i)=\pm1 \,
 \ee
as follows directly from (\ref{ortuv}) and (\ref{sympuv}). Moreover,
$\det(u_i)$ and $\det(v_i)$ do not depend on the specific choice of $u_i$
and $v_i$, since the determinant of the remaining orthogonal or sympectic
degree of freedom is $\det(O_\ell\otimes O_r)=\det(S_\ell\otimes S_r)=+1$
(since $r$ and $\ell$ are even).  We will henceforth use this choice of
$u_i$ and $v_i$ for the SF of $\mathcal U_i$.

\begin{thm} 
\label{prop:conjsym-strict-equiv}
If $r$ and $\ell$ are both even, two MPU
${\cal U}_{1,2}$, with $u_i$ and $v_i$ for their SF fulfilling 
(\ref{ortuv}) or (\ref{sympuv}), respectively,
are strictly equivalent under conjugation symmetry iff
\[
 \det(u_1)=\det(u_2), \quad \det(v_1)=\det(v_2)
\]
and moreover they correspond to the same case in Proposition
\ref{corconj}.
\end{thm}

\begin{proof}
(only if)
First, let us assume $\mathcal U_{1,2}$ correspond to different cases in
Prop.~\ref{corconj}.  Then, according to  Prop.~\ref{prop1conj} since
$y\otimes x=v^T v=\mathbb S\bar u u^\dagger \mathbb{S}^\dagger$ (with
$\mathbb S$ the swap), $x\otimes y$ is unitary and continuous in $u,v$,
and thus, we can choose $x$ and $y$ unitary and continuous as well, so
that they cannot change from symmetric to skew-symmetric.  

Let us thus now consider the case where $\mathcal U_{1,2}$ correspond to
the same case in Prop.~\ref{corconj}, and assume there is a path $u(p)$,
$v(p)$ that keeps the symmetry and that connects them continuously. Thus,
$u(p)$ and $v(p)$ fulfill Proposition \ref{prop1conj} with some $x(p)$,
$y(p)$. Since we can
write $x(p)\otimes y(p)=\bar u(p)u(p)^\dagger$, they can also be chosen
continuous. 
Using Lemma~\ref{lemma:conjclass-normalform-continuous},
we can write $x(p)=S_x(p)^T \tilde \Lambda_x(p) S_x(p)$ and $y(p)=S_y(p)^T
\tilde \Lambda_y(p) S_y(p)$ with $S_{x,y}(p)=S_{x,y}(p)^T$ unitary.  Since
$\det[\tilde\Lambda_{x,y}(p)]=1$, we have that
$\det[S_{x,y}(p)]^2$ is continuous as well.
Performing the gauge transformation
\begin{align*}
 u'(p)&=[S_x(p)\otimes S_y(p)]u(p),\\
 v'(p)&=v(p)[S_y(p)^\dagger\otimes S_x(p)^\dagger]
\end{align*}
we have that $u',v'$ generate the same MPU as $u,v$. Furthermore,
$\det[u'(p)]=\det[S_x(p)]^r \det[S_y(p)]^\ell \det[u(p)]$ is continuous
in $p$ since both $r$ and $\ell$ are even, and $u(p)$ is continuous.
But since $\overline{u'(p)}=[\tilde \Lambda_x(p)\otimes \tilde
\Lambda_y(p)] u'(p)$,  and $\det[\tilde\Lambda_{x,y}(p)]
=1$, we have that $\det[u'(p)]=\pm 1$. Thus, we must have
$\det[u'(0)]=\det[u'(1)]$. Finally, following the proof of
Prop.~\ref{corconj}, the gauge transformations $O_{x,y}(0)$,
$O_{x,y}(1)$ bringing $u'(0)$ and $u'(1)$ into the form
$u_1$ and $u_2$ fulfilling (\ref{ortuv}) or (\ref{sympuv}), 
respectively,
are orthogonal, and since $r$ and $\ell$ are even they
do not change the sign of the determinant, such that
$\det[u_1]=\det[u'(0)]=\det[u'(1)]=\det[u_2]$. The result for $v_1$ and
$v_2$ follows accordingly.
\\[0.5ex]
(if) Since both MPU are in the same case of Prop.~\ref{corconj}, we can
proceed as in Theorem~\ref{conjphs}. In Case I, since $u_1$ and $u_2$
are orthogonal and have the same determinant, they can be continuously
connected with $u(p)$ orthogonal. In Case II, we can always write
$u_{i},v_{i}$ in the form (\ref{upvpconj}), i.e.\ $u_i =
(\openone_{r\ell/4}\otimes R)\tilde u_i$ with some orthogonal $\tilde u_i$
with $\det(\tilde u_i)=\det(\tilde v_i)$, and correspondingly for $v_i$,
which can thus be connected by choosing a continuous path of orthogonal
matrices $\tilde u(p)$ and $\tilde v(p)$.
\end{proof}

\begin{cor}
\label{cor:equiv-under-conjugation}
${\cal U}_{1,2}$ are equivalent under conjugation symmetry iff they are in
the same case according to Proposition~\ref{corconj}.
\end{cor}

\begin{proof}
First, choose $u$ and $v$ such that they satisfy (\ref{ortuv}) or
(\ref{sympuv}), respectively.  By blocking once, we obtain new
$u^{(2)}_{11',22'}= v_{1'2}u_{11'}u_{22'}$ and
$v^{(2)}_{11',22'}=v_{1'2}$, so that
$\det(u^{(2)})=(\det(u)^2\det(v))^{r\ell}=1$ and
$\det(v^{(2)})=(\det(v))^{r\ell}=1$.  Moreover, it is straightforward to
check that $u^{(2)}$ and $v^{(2)}$ still satisfy (\ref{ortuv}) or
(\ref{sympuv}) (in the latter case with
$\Sigma_{\ell}^{(2)}=\Sigma_\ell\otimes\openone_r$ and
$\Sigma_{r}^{(2)}=\openone_\ell\otimes \Sigma_r$), and thus by applying
Thm.~\ref{prop:conjsym-strict-equiv}, we arrive at the conclusion.
\end{proof}

% ======================================================================
\subsection{\label{sec:necessary-cond-time-rev}
Necessary conditions under time reversal symmetry}

In the following, we will consider necessary conditions for
equivalence under symmetries which can be derived from the local
characterization of symmetries in Section~\ref{sec:sym-local-char}. Since the local
characterization of transposition symmetry,
Prop.~\ref{prop:local-char-transposition}, leads to no non-trivial
invariants (the phase in Eq.~(\ref{UUtransposetrick}) can be changed
smoothly while keeping the symmetry, cf.~the comment after the
proposition), the remaining case of interest is time reversal symmetry.

Let us start by recalling the local characterization of time reversal
symmetry in Prop.~\ref{PropChiral} for an MPU in SF,
Eq.~(\ref{UUdaggertrick}), which we here write for conciseness with
subscripts denoting the sites each operator acts on:
\begin{equation}
 v_{23}u_{12}u_{34}=
 \sigma x_{12} x_{34}^\dagger  v_{23}^\dagger
\label{UUdaggertrick-2}
\end{equation}
with $\sigma=\pm1$.   Note that $\sigma$ can be extracted directly from
$u$ and $v$:
\begin{lem}
For a time reversal invariant MPU in SF represented by $u$ and $v$,
$\sigma$ in Eq.~(\ref{UUdaggertrick-2}) is given by
\begin{equation}
\label{eq:sigma-from-trace}
\sigma = \tfrac{1}{d^2}\mathrm{tr}[\mathbb S_{12,34}v_{23}u_{12}u_{34}v_{23}]\ ,
\end{equation}
where $\mathbb S_{12,34}$ swaps sites $12$ with $34$. 
\end{lem}

\begin{proof}
This follows directly from $v_{23}u_{12}u_{34}v_{23}=\sigma\,
x_{12}x_{34}^\dagger$, together with the fact that 
$\mathrm{tr}[\mathbb S(a\otimes b)]=\mathrm{tr}[ab]$ and thus 
$\mathrm{tr}[\mathbb
S_{12,34}x_{12}x_{34}^\dagger]=\mathrm{tr}[xx^\dagger]=d^2$.
\end{proof}

Let us now show that the sign is gauge invariant and
cannot be changed by continuous deformations.

\begin{prop}
\label{prop:sigma-well-defined}
For an MPU $\mathcal U$ in SF represented by $u$ and $v$, $\sigma=\pm1$
given by Eq.~(\ref{UUdaggertrick-2}) is well defined (i.e., independent
of the gauge) and cannot be changed by continuous deformations of
$\mathcal U$ which keep time reversal symmetry.
\end{prop}

\begin{proof}  
It is straightforward to check that  the r.h.s.\ of
(\ref{eq:sigma-from-trace}) is invariant under gauge tranformations
$u\leadsto(x\otimes y)u$, $v\leadsto v(y^\dagger\otimes x^\dagger)$.
Also, (\ref{eq:sigma-from-trace}) implies that $\sigma$ is continuous (and
thus constant) under continuous changes of $u$, $v$ which keep time
reversal symmetry.
\end{proof}

This demonstrates that $\sigma$ is an invariant characterizing
inequivalent classes of MPUs.  In principle, there can be a different
$\sigma$ for each level of blocking, which all form invariants; for
instance, $u=i\openone$ and $v=\openone$ has $\sigma=-1$, while after
blocking $k=2$ sites, we have $u^{(2)}=-\openone$ and $v^{(2)}=\openone$
which have $\sigma^{(2)}=+1$. However, as we now show, the phases
$\sigma\equiv\sigma^{(1)}$ and $\sigma^{(2)}$ obtained in
Eq.~(\ref{UUdaggertrick-2}) after blocking $1$ and $2$ sites,
respectively, completely determine the phase for all other blockings.

\begin{prop}
\label{prop:sigma1-and-sigma2-equal}
Let $\mathcal U$ be a time reversal invariant MPU $\mathcal U$ in SF
described by $u$, $v$, and let the blocked MPU $\mathcal U_k$ be described
by $u^{(k)}$, $v^{(k)}$. Denote by $\sigma^{(k)}$ the phase obtained from
$u^{(k)}$ and $v^{(k)}$ in Eq.~(\ref{UUdaggertrick-2}). Then, 
\[
\sigma^{(k)} = \left\{\begin{array}{l@{\quad}l}
	\sigma^{(1)} & \mbox{if $k$ odd}\\
	\sigma^{(2)} & \mbox{if $k$ even\ .}
    \end{array}\right.
\]
\end{prop}

\begin{proof}
Because of Prop.~\ref{prop:sigma-well-defined}, we can choose the
representations $u^{(k)}$, $v^{(k)}$ at will.  For now, we choose
$u^{(2)}_{1234}=v_{23}u_{12}u_{34}$ and $v^{(2)}_{3456}=v_{45}$. Since
blocking preserves the symmetry, we have that 
\begin{equation}
\label{eq:defeq-sigma1}
v_{45}u_{34}u_{56}v_{45}=\sigma^{(1)} x_{34} x_{56}^\dagger
\end{equation}
(where we use site labels $3,4,5,6$) and
\begin{align*}
\sigma^{(2)} X_{1234} X_{5678}^\dagger 
&=
v^{(2)}_{3456} u^{(2)}_{1234} u^{(2)}_{5678} v^{(2)}_{3456}
\\
&=
v_{23}v_{45}v_{67}u_{12}u_{34}u_{56}u_{78}v_{45}
\\
&\stackrel{\eqref{eq:defeq-sigma1}}{=}
\sigma^{(1)}
v_{23}v_{67}u_{12}u_{78} x_{34} x_{56}^\dagger\ .
\end{align*}
From this, we infer (by equating the terms acting on sites $1234$ with
$X_{1234}$) that
\begin{equation}
\label{eq:unwind-dagger-trick}
\zeta
[(\openone\otimes v\otimes\openone)(u\otimes x)]^\dagger = 
(\openone\otimes v\otimes\openone)(x^\dagger\otimes u)\ ,
\end{equation}
where $\zeta=\sigma^{(2)}/\sigma^{(1)}=\sigma^{(1)}\sigma^{(2)}$. Now
consider $k$ blocked sites, where we choose 
\begin{align*}
u'\equiv u^{(k)} 
&= 
(\openone\otimes v^{\otimes k-1}\otimes\openone)(u^{\otimes k})
\\
v'\equiv v^{(k)} 
&=
\openone^{\otimes (k-1)} \otimes v\otimes \openone^{\otimes (k-1)}
\end{align*}
Eq.~(\ref{UUdaggertrick-2}) with $k$ blocked sites now reads
\[
X\otimes Y = (\openone^{\otimes k}\otimes v'\otimes\openone^{\otimes k})(u'\otimes u')
(\openone^{\otimes k}\otimes v'\otimes\openone^{\otimes k})\ ,
\]
where $\sigma^{(k)}$ can be obtained from $Y=\sigma^{(k)}X^\dagger$ (this
relation must hold with some phase as blocking preserves the symmetry). By
using (\ref{eq:defeq-sigma1}) on the two middle sites, we find that 
\begin{align*}
X&=
(\openone\otimes v^{\otimes(k-1)}\otimes\openone)
	    (u^{\otimes(k-1)}\otimes x)\ ,
\\
Y&=\sigma^{(1)}
(\openone\otimes v^{\otimes(k-1)}\otimes\openone)
	    (x^\dagger\otimes u^{\otimes(k-1)})\ .
\end{align*}
We can now sequentially apply (\ref{eq:unwind-dagger-trick}) to relate
$X$ and $Y$: Consider
\begin{align*}
Y_s &= \sigma^{(1)}\zeta^s
    (\openone^{\otimes (2s+1)}\otimes  v^{\otimes(k-1-s)}\otimes\openone)
    \times
\\
&\qquad 
    ( u^{\dagger\otimes s}\otimes x^\dagger \otimes  u^{\otimes (k-1-s)})
    (\openone\otimes  v^{\dagger\otimes s}\otimes\openone^{\otimes (2k-2s-1)})
\end{align*}
for $s=0,\dots,k-1$,
and note that by applying (\ref{eq:unwind-dagger-trick}), we can show that
$Y_s=Y_{s+1}$.  We thus have that
\[
Y=Y_0=\dots=Y_{k-1}=\sigma^{(1)}\zeta^{k-1}X^\dagger\ ,
\]
and thus $\sigma^{(k)} = \sigma^{(1)}\zeta^{k-1}$ as claimed.
\end{proof}

\begin{cor}
Consider two MPU $\mathcal U_1$ and $\mathcal U_2$ in SF. For $\mathcal
U_1$ and $\mathcal U_2$ to be strictly equivalent, it is necessary that
both $\sigma^{(1)}$ and $\sigma^{(2)}$
(cf.~Prop.~\ref{prop:sigma1-and-sigma2-equal}) for the two MPU are equal.
For $\mathcal U_1$ and $\mathcal U_2$ to be equivalent, it is necessary
that $\sigma^{(2)}$ for the two MPU is equal.
\end{cor}

% ======================================================================
\subsection{Connection to the classification of phases of Matrix Product
States}

Two MPUs are equivalent if they can be connected by a continuous path of
MPUs. A necessary condition for the existence of such a path is the
existence of a continuous path of normal MPVs,
cf.~Prop.~\ref{prop:normal-tensor}, without requiring unitarity. This is
equivalent to asking whether the corresponding injective MPS are in the
same phase as ground states of 1D Hamiltonians. This question has been
studied
previously~\cite{pollmann:1d-sym-protection-prb,chen:1d-phases-rg,schuch:mps-phases},
and it has been found that this relates to the classification of the gauge
transformations
\begin{equation}
\label{eq:mps-class-gauge-trafo}
\mathcal S[\mathcal U] \propto G \mathcal U G^{-1}\ ,
\end{equation}
which represent the symmetry action
locally~\cite{pollmann:1d-sym-protection-prb,chen:1d-phases-rg,schuch:mps-phases}.
In particular, for the three cases considered in this paper, it can be
easily checked that the classification of $G$ is as follows:
\begin{enumerate}
\item \emph{Conjugation:} $G\bar{G}=\pm \openone$.
\item \emph{Time reversal:} $G\bar{G}=\pm \openone$.
\item \emph{Transposition:} No nontrivial phases.
\end{enumerate}
Since $G$ is invariant under blocking and adding ancillas, it gives
necessary conditions for equivalence of MPUs,
cf.~Def.~\ref{def:equivalent-symmetry}. The symmetry
in the two first cases corresponds to time reversal symmetry in the MPS
classification~\cite{pollmann:1d-sym-protection-prb}, where the
condition is obtained by applying $\mathcal S[\mathcal S[\mathcal
U]]\propto \mathcal U$.  Note
that by a suitable choice of gauge, $G$ can be chosen
unitary~\cite{cirac:mpdo-rgfp},
and thus the cases $G\bar{G}=\pm\openone$ correspond to symmetric and
skew-symmetric unitary matrices, which form two disconnected equivalence
classes
(cf.~Corollary~\ref{corconj}).  The transposition symmetry corresponds to
an on-site $\mathbb Z_2$ symmetry, which has no non-trivial projective
representations and thus no non-trivial MPS
phases~\cite{chen:1d-phases-rg,schuch:mps-phases}.

\begin{prop}
\label{prop:mps-class-eq-mpu-class-for-conjugation}
For an MPU invariant under conjugation symmetry, the sign $G\bar
G=\pm\openone$ in the MPS classification of phases,
Eq.~\eqref{eq:mps-class-gauge-trafo}, corresponds to the two cases in
Prop.~\ref{corconj}, i.e., the two possible MPU phases under equivalence.
\end{prop}

\begin{proof}The statement follows immediately from the fact that the two
cases in Prop.~\ref{corconj} label the two phases under equivalence under
conjugation (Corollary~\ref{cor:equiv-under-conjugation}), together with
the fact that equality of the sign $G\bar{G}=\pm\openone$ is a necessary
condition for being in the same phase. The fact that Case I corresponds to
$G\bar G=+\openone$ can be checked immediately by setting $u=v=\openone$.
\end{proof}

\begin{prop}
\label{prop:dagger-mps-class-and-mpu-class-equal}
For an MPU in SF invariant under time reversal, the sign $G\bar
G=\pm\openone$ in the MPS classification of phases,
Eq.~\eqref{eq:mps-class-gauge-trafo}, corresponds to the sign obtained in
Prop.~\ref{PropChiral} after blocking $k=2$ sites, i.e., $\sigma^{(2)}$ in
Prop.~\ref{prop:sigma1-and-sigma2-equal}.
\end{prop}

\begin{proof} Consider $\mathcal U$ in SF. $\mathcal U$ is normal and thus
an injective MPS tensor, and thus we have
\[
\includegraphics[width=5cm]{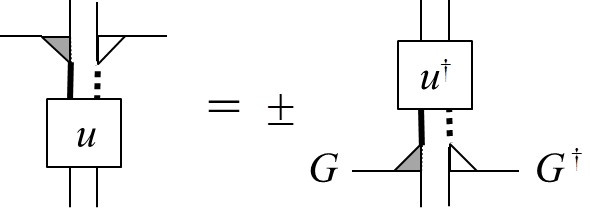}
\]
where with an appropriately chosen gauge, $G$ is unitary.  Using this, we
have that 
\[
\includegraphics[width=\columnwidth]{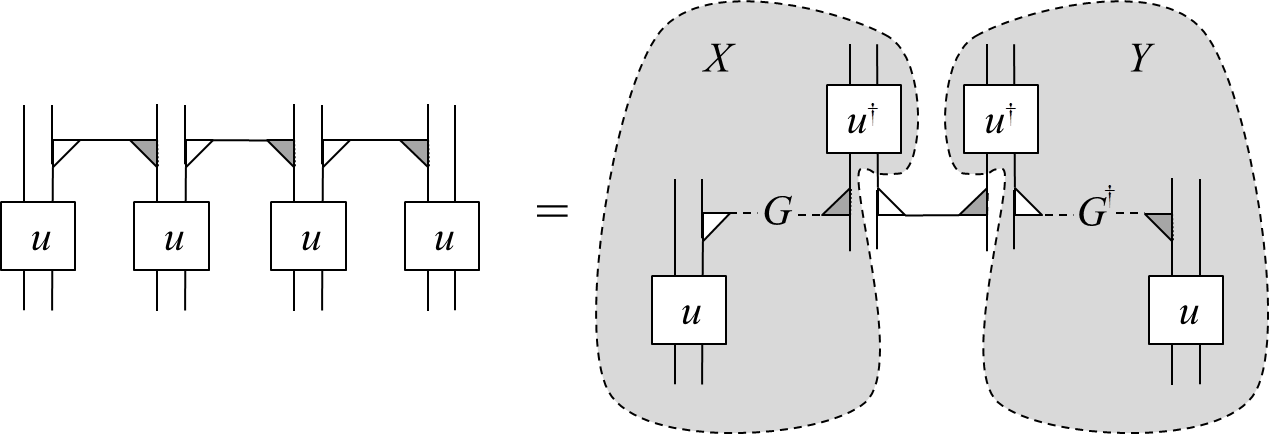}
\]
Using that $\bar G=\pm G^{\dagger}$, we can now immediately verify that
$X^\dagger = \pm Y$, and thus, $\sigma^{(2)}=\pm 1$ as claimed.
\end{proof}

\section{\label{sec:sym-examples}Symmetries: Examples}

In this Section, we analyze some relevant examples. First, we give an
example for each of the phases that have been identified in the previous
section. Then, we consider two MPUs that represent toy models of
topological insulators, and study if they are connected among them and to
the identity operator under different symmetries.

\subsection{Conjugation symmetry: Completeness of classification}

Let us start by providing examples for all possible phases under
conjugation.

\begin{example}[All phases under conjugation] 
\quad\newline
1. Let $u_+=v_+=\openone_{r\ell}$, and
$(u_-)=(v_-)=\mathrm{diag}(-1,1,\dots,1)$. Then, 
    \begin{itemize}
    \item in case at least one of $r$, $\ell$ is odd, 
    the MPUs given by $(u_+,v_+)$ and by $(u_+,v_-)$ lie in the two
    classes of Theorem~\ref{conjphs}.
    \item in case both $r$ and $\ell$ are even, the MPUs given by
    $(u_\pm,v_\pm)$ (with all four sign choices) lie in the four classes
    of Theorem~\ref{prop:conjsym-strict-equiv} for which $u$ and $v$ fall
    into Case I of Prop.~\ref{corconj}.
\end{itemize}
2. Let $u'_\pm=(\openone_{r\ell/4}\otimes R)u_{\pm}$ and
$v'_\pm=(\openone_{r\ell/4}\otimes R)v_{\pm}$, with $R=\frac{1-i}{2}
[\Id_4 + i(Y\otimes Y)]$, as in
Corollary~\ref{cor:sympl-to-orthogonal}. Then, the MPUs given by
$(u'_\pm,v'_\pm)$ (with all four sign choices) lie in the four classes of
Theorem~\ref{prop:conjsym-strict-equiv} for which $u$ and $v$ fall into
Case II of Prop.~\ref{corconj}.
\end{example}
\noindent
The verification is straightforward by using that $u_\pm$, $v_\pm$ are
unitary, $u'_\pm$, $v'_\pm$ fall into Case II following
Corollary~\ref{cor:sympl-to-orthogonal},
$\det(u_\pm)=\det(v_\pm)=det(u_+v_\pm)=\pm1$, and $\det(R)=-1$ and thus
$\det(u'_{\pm})=\det(v'_{\pm})=\pm (-1)^{r\ell/4}$.

Note that e.g.\ for the MPU described by $u=v=R$, the corresponding
$\mathcal U$ is
\[
\mathcal U_{\ell r}^{ud}=
\tfrac{1-i}{2}((\sigma_y)^{\ell}(-i\sigma_y)^{r})_{ud}
\]
with $\ell,r,u,d=0,1$ denoting the left, right, up, and down indices,
respectively.  This $\mathcal U$ needs to be blocked once to obtain the
SF; then, it is straightforward to check that conjugation maps to
a $G=\sigma_y$ gauge, which satisfies $G\bar{G}=-\openone$,
i.e., we are in the non-trivial MPS phase (as shown in
Prop.~\ref{prop:mps-class-eq-mpu-class-for-conjugation}).  It can be
easily seen that this generalized to all above examples, as the virtual
degrees of freedom relating to $u_\pm$ and $v_\pm$ transform trivially
under conjugation: Their MPU tensor is given by $\delta_{u=d=\ell}x_{\ell
r}$ with $x=u_\pm,v_\pm$ and thus real.

Let us also briefly consider these examples in the light of the other
symmetries.  Clearly, the MPUs given by $(u_\pm,v_\pm)$ have time reversal
and transposition symmetry, and lie in the trivial class with regard to the
former.  In the second case, it is easy to check that only the MPU with
$(u_+',v_+')$ has time reversal and transposition symmetry (the others
violate it already on two sites); the latter has 
$\sigma^{(1)}=\sigma^{(2)}=1$ in
Prop.~\ref{prop:sigma1-and-sigma2-equal}.

\subsection{Time-reversal symmetry}

Let us now provide examples for the four classes which we identified under
time reversal symmetry, namely $\sigma^{(1)}=\pm1$ and
$\sigma^{(2)}=\pm1$, Prop.~\ref{prop:sigma1-and-sigma2-equal}.

\begin{example}[CZX]
\label{example:czx}
The CZX MPU~\cite{chen:2d-spt-phases-peps-ghz} is defined by
\begin{align*}
u &= CZ(\sigma_x\otimes \sigma_x)\\
v &= CZ\ ,
\end{align*}
where $CZ=\mathrm{diag}(1,1,1,-1)$ is the controlled-Z gate. The
corresponding $\mathcal U$ is
\[
\mathcal U_{\ell r}^{ud} = (\sigma_x)_{ud}\delta_{u,\ell}H_{\ell,r}
\]
with $\ell,r,u,d=0,1$ denoting the left, right, up, and down indices,
respectively, and $H=\tfrac{1}{\sqrt{2}}
\left(\begin{smallmatrix}1&1\\1&-1\end{smallmatrix}\right)$ the Hadamard
gate.
\end{example}
\noindent
It is straightforward to check that this example has $\sigma^{(1)}=-1$ and
$\sigma^{(2)}=-1$ in Prop.~\ref{prop:sigma1-and-sigma2-equal}. By blocking
two sites of the MPS description, we obtain SF (and thus injectivity), and
we find that the virtual gauge transformation arising from time reversal
is $G=\sigma_y$, corresponding to the non-trivial phase (as shown in
Prop.~\ref{prop:dagger-mps-class-and-mpu-class-equal}).

Note that this example also has the other two symmetries. For conjugation
symmetry and strict equivalence, it is in the class which satisfies Case I
and has $\det(u)=\det(v)=-1$ in Theorem~\ref{prop:conjsym-strict-equiv}.
Under blocking, it thus falls into the trivial class, which can also be
verified using the MPS representation (the gauge is $G=\openone$).

The trivial case under time reversal can just be realized by setting
$u=v=\openone$.  The other non-trivial cases of $\sigma^{(1)}$ and
$\sigma^{(2)}$ can be constructed by modifying the trivial and CZX
example, respectively.

\begin{example}
\label{ex:dagger-all-but-czx}
1. The MPU $u=v=\openone$ realizes the trivial case 
$\sigma^{(1)}=\sigma^{(2)}=1$.\\[0.5ex]
2. The modified trivial MPU $u=i\openone$, $v=\openone$  realizes the case
$\sigma^{(1)}=-1$, $\sigma^{(2)}=1$.\\[0.5ex]
3. The modified CZX MPU $u=i\,CZ(\sigma_x\otimes \sigma_x)$, 
$v=CZ$  realizes the case $\sigma^{(1)}=1$, $\sigma^{(2)}=-1$.
\end{example}

Finally, these examples can be extended to any even dimension $\ell=r=d$
by tensoring the CZX example with $\openone_{d/2}$; clearly, the argument
of Example~\ref{ex:dagger-all-but-czx} still applies.

\subsection{Examples motivated by topological insulators}

In the following, we will discuss a class of examples motivated by
topological insulators, where the edge currents for each component of the
spin are chiral. Thus, a toy model describing the dynamics of the edge
states would correspond to a spin $1/2$ particle; if it is up, it moves to
the right and if it is down, to the left. Here we will consider many-body
versions of such model.

% ============================================================
\subsubsection{The Examples}

We consider here three examples where the local dimension is $\tilde
d=d^2$, i.e.  it corresponds to having two spins per site. We define the
shift operator \cite{gross:index-theorem}, $T^{(N)}$, acting on
$H_{d}^{\otimes N}$ as
 \be
 T^{(N)} |n_1,\ldots,n_N\rangle = |n_N,n_1,\ldots,n_{N-1}\rangle.
 \ee
That is, it shifts product states to the right. Its adjoint, $T^{(N)\dagger}$ shifts to the left. We also define
 \begin{subequations}
 \label{threeMPU}
 \bea
 U_1^{(N)} &=& (\Id\otimes\Id)^{\otimes N},\\
 U_2^{(N)} &=& T^{(N)\dagger}\otimes T^{(N)},\\
 U_3^{(N)} &=& T^{(N)}\otimes T^{(N)\dagger}.
 \eea
 \end{subequations}
The physical meaning of the second (third) unitaries is simple: In a site,
the spin on the right is shifted to the right (left) and the one on the
left, to the left (right). We illustrate $U_2^{(N)}$ and $U_3^{(N)}$ in
Fig.~\ref{fig1}. Note that $U_3^{(N)}=\mathbb S^{\otimes N} U_2^{(N)}
\mathbb S^{\otimes
N}$.

\begin{figure}
 {\includegraphics[height=8em]{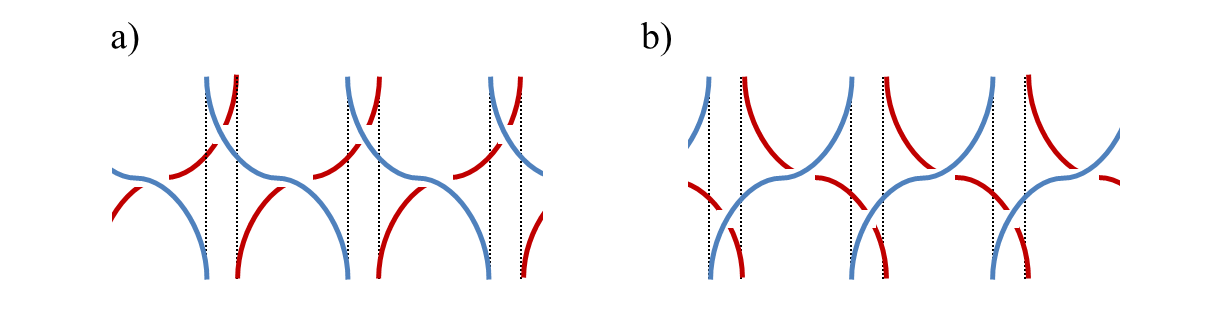}}
 \caption{\label{fig1}
Graphical representation of: a) $U_2^{(N)}$, and b) $U_3^{(N)}$. Solid lines indicate how different spins are shifted.}
\end{figure}

The matrices $u_i,v_i$ ($i=1,2,3$) generating $U_{i}^{(N)}$ can be easily derived from Fig. \ref{fig1}. We have
\begin{subequations}
\label{eq:SF_u1_u3}
\begin{align}
 u_1&=\Id\otimes \Id, \quad& v_1 &= \Id\otimes \Id\ ,
\\
 u_3&=\mathbb S\ , \quad& v_3&=\mathbb S\ .
\end{align}
\end{subequations}
For the ones corresponding to $U_2^{(N)}$ we must block two sites, which
yields the blocked matrices
\begin{equation}
\label{eq:uv2_U2}
\begin{aligned}
u_2^{(2)} &=  \Id\otimes \mathbb S\otimes \Id\ ,
\\
v_2^{(2)}&= (\mathbb S\otimes \mathbb S)(\Id\otimes \mathbb S\otimes \Id)\ .
\end{aligned}
\end{equation}
and correspondingly
\begin{equation}
\label{eq:uv2_U3}
\begin{aligned}
u_3^{(2)} &=  (\Id\otimes \mathbb S\otimes \Id)(\mathbb S\otimes \mathbb S)\ ,\\
v_3^{(2)} & = (\Id\otimes \mathbb S\otimes \Id)\ ,
\end{aligned}
\end{equation}

All three MPUs in Eq.~(\ref{threeMPU}) are generated by tensors which have
${\rm ind}=0$ (this follows from Theorem \ref{IndexTh} and the fact that
$T^{(N)}$ has index ${\rm ind}=\log_2(d_0)$ whereas $T^{(N)\dagger}$ has
${\rm ind}=-\log_2(d_0)$).  They have conjugation symmetry, as well as 
the symmetries (\ref{SSymmetries}) with $Q=\mathbb S$, this is, time reversal or
transpose combined with an exchanging the two species of sites.

% ============================================================
\subsubsection{Transformation with Swaps}

In the following we will analyze whether the three MPU in Eq.~(\ref{threeMPU})
are in the same phase with respect to the three symmetries mentioned
above, this is, conjugation as well as the symmetries of
Eq.~(\ref{SSymmetries}) with $Q=\mathbb S$. In the latter case, following Section
\ref{othersymmetries}, it is convenient to define new MPU, $\tilde {\cal
U}=  {\cal U}\mathbb S$, generating 
\begin{subequations}
 \label{threeMPU2}
 \begin{align}
 \tilde U_1^{(N)} &= \mathbb S^{\otimes N},\\
 \tilde U_2^{(N)} &= \mathbb S^{\otimes N}(T^{(N)\dagger}\otimes T^{(N)}),\\
 \label{tildeU3}
 \tilde U_3^{(N)} &= \mathbb S^{\otimes N}(T^{(N)}\otimes T^{(N)\dagger})=
\mathbb S^{\otimes N}\tilde U_2^{(N)}\mathbb S^{\otimes N}.
  \end{align}
 \end{subequations}

\begin{lem}
\label{lemma:sym-trafo-swap}
Any two MPU $U^{(N)}_i$ and $U^{(N)}_j$ in (\ref{threeMPU}) are in the
same phase under (\ref{SSymmetries-dagger}), 
(\ref{SSymmetries-transpose}) (both with $Q=\mathbb S$), and conjugation,
respectively, iff the corresponding MPU
$\tilde U^{(N)}_i$ and $\tilde U^{(N)}_j$ in (\ref{threeMPU}) 
 are in the same phase under time reversal, 
transposition, and conjugation, respectively.  
\end{lem}

\begin{proof}
The first two symmetries follow from the discussion in
Sec.~\ref{othersymmetries}.  The last follows from the fact that $\tilde
U^{(N)}=\mathbb S^{\otimes N}$ is invariant under conjugation iff $U^{(N)}$ is,
together with the the fact that any continuous path $U(x)^{(N)}$ with
conjugation symmetry gives rise to a continuous path $\tilde
U(x)^{(N)}=\mathbb S^{\otimes N}U(x)^{(N)}\mathbb S^{\otimes N}$ with conjugation
symmetry, and vice versa. Note that in all cases, translation symmetry is
kept (though the transformation might break or restore the SF).
\end{proof}

Thus, for analyzing those symmetries we can restrict ourselves to the form
(\ref{threeMPU2}).

The SF of tensors generating the MPUs $\tilde U_1^{(N)}$ and $\tilde
U_2^{(N)}$ have
 \begin{subequations}
 \label{SFu1u3}
 \bea
 \tilde u_1 &=& \mathbb S, \quad \tilde v_1 = \Id\otimes\Id,\\
 \tilde u_2 &=& \Id\otimes\Id, \quad \tilde v_2=\mathbb S.
 \eea
 \end{subequations}
(We will not need the one corresponding to the third.)
Note that all three MPUs have the same index
and spectrum, and thus just by looking at spectral properties we cannot
distinguish their phases under time reversal (as the spectrum must be
real and thus cannot be changed continuously). The fact that $\tilde
U_2^{(N)}$ and $\tilde U_3^{(N)}$ have the same spectrum is clear since
they are related by a unitary transformation, the global swap
(\ref{tildeU3}). The fact that $\tilde U_1^{(N)}$ and $\tilde U_2^{(N)}$
have the same spectrum is obvious from their SF (\ref{SFu1u3}).

% ============================================================
\subsubsection{Phases of $U_2^{(N)}$ vs.\ $U_3^{(N)}$}

Let us now compare $U_2^{(N)}$ vs.\ $U_3^{(N)}$ under the three
symmetries.  We start with conjugation. In order to compare the two, we
have to block two sites and thus consider the SF given in
Eqs.~(\ref{eq:uv2_U2}) and (\ref{eq:uv2_U3}), respectively. First, since
$u_i^{(2)}$, $v_i^{(2)}$ are real, they transform trivally under
conjugation and we are thus in Case I.  We then have that
\begin{align*}
\det(u_2^{(2)})&=\det(v_3^{(2)})=\det(\mathbb S)^{d^2}
\\
\det(u_3^{(2)})&=\det(v_2^{(2)})=(\det(\mathbb S)^{d^2})^3=\det(\mathbb S)^{d^2}\ ,
\end{align*}
where we have used that $\det(\mathbb S)=\pm 1$.  If $d$ is odd,
Theorem~\ref{conjphs} applies and 
$\det(u_2^{(2)})=\det(v_2^{(2)})$ and thus 
$\det(u_2^{(2)}v_2^{(2)})=1$; if $d$ is even, 
Theorem~\ref{prop:conjsym-strict-equiv} applies with
$\det(u_i^{(2)})=\det(v_i^{(2)})=1$.  We thus find: 

\begin{prop}
\label{prop:U2-U3-trivial-conjugation}
After blocking $2$ sites, $U_2^{(N)}$ and $U_3^{(N)}$ are both in the
trivial phase according to conjugation.
\end{prop}

Now, let us move to the other symmetries.
\begin{prop}
\label{prop:U2-U3-trivial-tr-transpose}
The MPU $\tilde U_{2,3}^{(N)}$ are in the same phase with respect to both
time reversal and transposition symmetry.
\end{prop}

\begin{proof}
We will show that there exists a continuous path of tensors ${\cal U}(x)$,
$x\in [0,1]$, generating $\tilde U(x)^{(N)}$, such that $\tilde
U(0)^{(N)}=\tilde U_2^{(N)}$ and $\tilde U(1)^{(N)}=\tilde U_3^{(N)}$.
Let us denote by ${\cal U}$ the tensor generating $\tilde
U_2^{(N)}$. We choose a unitary operator $\mathbb S(x)$ such that $\mathbb S(0)=\Id$ and
$\mathbb S(1)=\mathbb S$, and define ${\cal U}(x)=\mathbb S(x){\cal
U}\mathbb S(x)^\dagger$, where $\mathbb S(x)$ acts on
the physical indices of $\mathcal U$. It is obvious that this generates a
self-adjoint MPU.  Together with (\ref{tildeU3}), this yields the
statement of the proposition. A similar construction with $\mathbb S(x)^\dagger\to
\mathbb S(x)^T$ works for the transposition symmetry.  
\end{proof}

% ============================================================
\subsubsection{Phases of $U_1^{(N)}$ vs.\ $U_2^{(N)}$}

We now turn towards $U_1^{(N)}$ vs.\ $U_2^{(N)}$.
Let us start with conjugation symmetry.

\begin{prop}
\label{prop:eqiv-U1-U2-conjugate}
$U_1^{(N)}$ and $U_2^{(N)}$ are in different phases if $d=4k+2$,
$k=0,1,2,\dots$, and otherwise in the same phase, where the phase of
$U_1^{(N)}$ is (by construction) trivial.
All phases are in Case I, and thus 
$U_1^{(N)}$ and $U_2^{(N)}$ coincide under equivalence (i.e.\ blocking and
adding ancillas).
\end{prop}

\begin{proof}
Following Lemma~\ref{lemma:sym-trafo-swap}, we can consider $\tilde
U_1^{(N)}$ and $\tilde U_2^{(N)}$ instead, and consider the SF
(\ref{SFu1u3}). Again, both fall into Case I. We have that
\[
\det(\tilde u_1)=\det(\tilde v_2)=\det(\mathbb S)=(-1)^{d(d-1)/2}\ ,
\]
while $\det(\tilde v_1)=\det(\tilde u_2)=1$. Thus, for $d$ odd,
$\det(\tilde u_1)\det(\tilde v_1)=\det(\tilde u_2)\det(\tilde
v_2)=(-1)^{d(d-1)/2}$ and thus
$\tilde U_1^{(N)}$ and $\tilde U_2^{(N)}$ are in the same phase
(Theorem~\ref{conjphs}); it is non-trivial if $d=4k+3$.  If $d$ is even,
$\det(\tilde u_1)=\det(\tilde v_1)=1$ and $1=\det(\tilde u_2)=\det(\tilde
v_2)$ iff $d(d-1)/2$ is even, this is $d$ is divisible by $4$: $\tilde
U_1^{(N)}$ and $\tilde U_2^{(N)}$ are in the same phase if $d=4k$, and in
different phases if $d=4k+2$ (Theorem~\ref{prop:conjsym-strict-equiv}).
\end{proof}

As already indicated in the proof, for $d=4k+2$ the two phases of $\tilde
U_1^{(N)}$ and $\tilde U_2^{(N)}$ are both non-trivial, while in the other
cases, the joint phase which both MPUs share is only non-trival if
$d=4k+3$.  Since we did not block here, we in principle cannot translate
this back to a statement about the nature of the phase of $U_2^{(2)}$.
However, since $U_1^{(N)}$ is by construction in the trivial phase, and we
know how $U_2^{(N)}$ relates to $U_1^{(N)}$, this nevertheless allows to
infer whether the phase of $U_2^{(N)}$ is non-trivial.

Let us now turn towards the other two symmetries. It can be checked
straightforwardly from (\ref{SFu1u3}) that $\sigma^{(1)}=\sigma^{(2)}=1$
in Prop.~\ref{prop:sigma1-and-sigma2-equal}, i.e., we obtain no separation
of $\tilde U_1^{(N)}$ and $\tilde U_2^{(N)}$ under time reversal.  Indeed,
it is possible to show that $\tilde U_1^{(N)}$ and $\tilde U_2^{(N)}$ are
in the same phase once we allow for ancillas:

\begin{prop}
\label{prop:U1-U2-equiv-ancillatrick}
If we add one ancilla of dimension $d$ per site, $\tilde U_1^{(N)}$ and
$\tilde U_2^{(N)}$ are stricly equivalent both under time reversal and
transposition symmetry.
\end{prop}

\begin{figure}[t]
\includegraphics[width=\columnwidth]{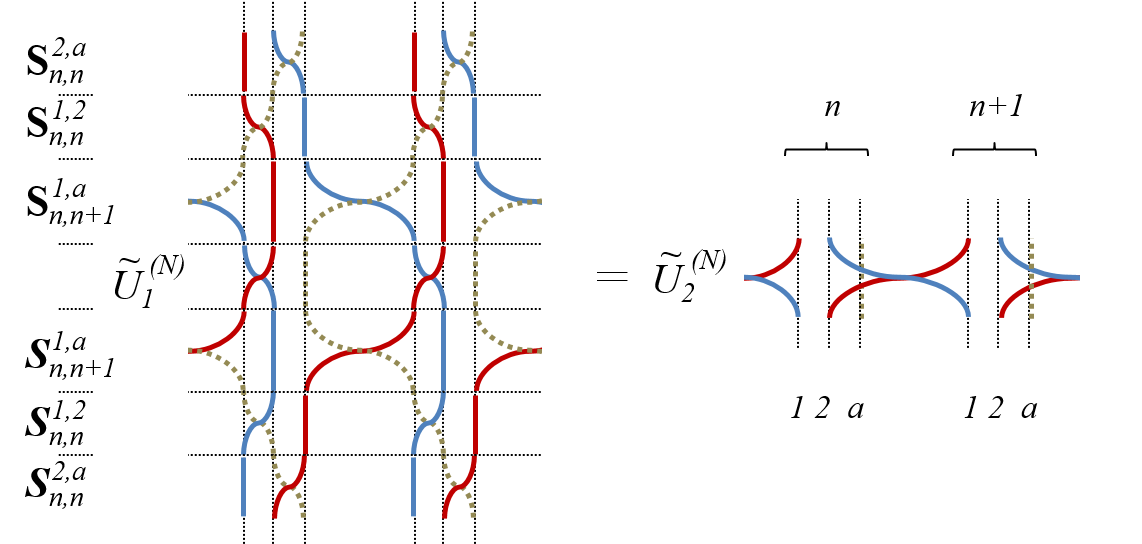}
\caption{\label{fig:TR-ancilla}
Proof of Prop.~\ref{prop:U1-U2-equiv-ancillatrick}: By adding an ancilla,
$\tilde U_1^{(N)}$ and $\tilde U_2^{(N)}$ can be connected through a
path with time reversal or transposition symmetry.
}
\end{figure}

\begin{proof}
The proof is illustrated in Fig.~\ref{fig:TR-ancilla}: By adding an
ancillas and conjugating $\tilde U_1^{(N)}$ with the three swaps shown, we
obtain $\tilde U_2^{(N)}$.  We can now choose a continuous interpolation
$\mathbb S(x)$ such that $\mathbb S(0)=\openone$ and $\mathbb
S(1)=\mathbb S$. By conjugating with $\mathbb S(x)$ and $\mathbb
S(x)^\dagger$ or $\mathbb S(x)^T$ instead of the swaps, we obtain a
continuous time reversal or transposition symmetric path which connects
$\tilde U_1{(N)}$ and $\tilde U_2{(N)}$, as claimed.
\end{proof}

\subsubsection{Phases of $U_1^{(N)}$ vs.\ $U_3^{(N)}$}

Finally, let us consider the phase of $U_1^{(N)}$ vs.\ $U_3^{(N)}$.  It is
immediate to see that all results from the preceding section can be
transferred, e.g.\ by observing that there is no relevance to the ordering
of the two spins per site, and swapping them exchanges $U_2^{(N)}$ with
$U_3^{(N)}$ (and $\tilde U_2^{(N)}$ with $\tilde U_3^{(N)}$), or by noting
the correspondence between Eqs.~(\ref{eq:SF_u1_u3}) and (\ref{SFu1u3}).
Note that in this case, however, the proof analogous to that of
Prop.~\ref{prop:eqiv-U1-U2-conjugate} will additionally allow to determine
the phase of $U_2^{(N)}$ under strict equivalence, i.e., without any
blocking. (If one is only
interested in the results after blocking two sites, one can alternatively
also transfer the results of the preceding section from $U_2^{(N)}$ to
$U_3^{(N)}$ using Propositions~\ref{prop:U2-U3-trivial-conjugation}
and~\ref{prop:U2-U3-trivial-tr-transpose}.)

% ============================================================

\begin{acknowledgements}
This project has been supported by the EU through the ERC grants WASCOSYS
(No.~636201), GAPS (No.~648913), and QUTE (No.~647905). 
DPG acknowledges support by MINECO through grants MTM2014-54240-P and
ICMAT Severo Ochoa SEV-2015-0554, and by Comunidad de Madrid through grant
QUITEMAD+-CM, ref. S2013/ICE-2801.

\end{acknowledgements}

\appendix

\section{MPU and QCA}

In this appendix we recall from \cite{schumacher:qca} the definition of
QCA and prove that 1D QCAs and MPUs are exactly the same objects. We also
show that the index defined here coincides (up to a log) with the original
index of QCA introduced in Ref.~\cite{gross:index-theorem}.

In order to define QCA we have to start with the quasi-local algebra of observables in a 1D spin chain. Let us take  a finite $\Lambda\subset \mathbb{Z}$ and denote by $\mathcal{A}_\Lambda$ the algebra of observables of $\Lambda$, that is $M_d^{\otimes \Lambda}$. If $\Lambda \subset \Lambda'$ one can identify $\mathcal{A}_\Lambda$ as a subset of $\mathcal{A}_{\Lambda'}$ by tensoring each element with the identity in the sites of $\Lambda'\backslash\Lambda$. This allows to define the inductive limit
\begin{equation}
\mathcal{A}_{\rm loc}=\bigcup _{\Lambda}^\infty \mathcal{A}_\Lambda
\end{equation}
After completion with the norm topology one obtains a $C^*$-algebra ${\mathcal{A}}$ called the quasi-local algebra of observables.

A 1D QCA is simply an automorphism $\omega$ of ${\mathcal{A}}$, which commutes with the translation operator and for which there exists a finite subset $\mathcal{N}\subset \mathbb Z$ so that $\omega(\mathcal{A}_{\Lambda})\subset \mathcal{A}_{\Lambda +\mathcal{N}}$ for all finite $\Lambda\subset \mathbb Z$. Wlog one can assume that $\mathcal{N}$ is an interval $[-R,R]\cap\mathbb Z$.

Let us show first that any MPU given by tensor $\mathcal{U}$ defines a QCA. For that we consider
$\omega_{\mathcal U}: \mathcal{A}_{\rm loc}\rightarrow \mathcal{A}_{\rm loc}$
given by
\begin{equation}\label{eq:appendix-1}
\omega_{\mathcal U} (A)=\lim_N U^{(N)}A_L U^{(N)\dagger}.
\end{equation}
Note that $\omega_{\mathcal U}$ is well defined since the limit is eventually constant thanks to the results proven in the main text. Moreover,  $\omega_{\mathcal U}$ preserves the norm and all operations of the $*$-algebra $\mathcal{A}_{\rm loc}$, commutes with the translation operator and verifies that $\omega_{\mathcal U}(\mathcal{A}_{\Lambda})\subset \mathcal{A}_{\Lambda +\mathcal{N}}$ for $\mathcal N = [-4D^4,4D^4]\cap \mathbb Z$. It can be then uniquely extended to a 1D QCA as defined above.

Conversely, by the main result of \cite{schumacher:qca}, after blocking a finite number of sites, any 1D QCA  $\omega$ restricted to $\mathcal{A}_{\rm loc}$ is exactly of the form (\ref{eq:appendix-1}), where the unitaries $U^{(N)}$ are given by  depth-two circuits of nearest neighbor unitaries $u$ and $v$. This gives exactly the standard form (SF) of an MPU $\mathcal{U}$ (\ref{StandardForm}) and $\omega=\omega_{\mathcal{U}}$ for such MPU.

Let us finish this appendix by showing that the index defined here
coincides with the one defined in \cite{gross:index-theorem}. We will use
the identification just made between QCA and MPU. Let us recall the
original definition of \cite{gross:index-theorem}. For that we group sites
together so that the MPU $\mathcal{U}$ is simple. Then clearly, for $x\in
\mathbb Z$,
$$\omega(\mathcal{A}_{2x}\otimes \mathcal{A}_{2x+1})\subset
\left( \mathcal{A}_{2x-1}\otimes \mathcal{A}_{2x}\right)\otimes
\left(\mathcal{A}_{2x+1}\otimes \mathcal{A}_{2x+2}\right).$$

Note that all algebras $\mathcal{A}_{z}$, $z\in \mathbb Z$, are isomorphic to $M_d$.

In order to define the index, one needs to identify the support algebra of
$\omega(\mathcal{A}_{2x}\otimes \mathcal{A}_{2x+1})$ on
$\mathcal{A}_{2x-1}\otimes \mathcal{A}_{2x}$, that is, the smallest
$C^*$-subalgebra $\mathcal{R}_{2x}$ of $\mathcal{A}_{2x-1}\otimes
\mathcal{A}_{2x}$ so that $\omega(\mathcal{A}_{2x}\otimes
\mathcal{A}_{2x+1})\subset \mathcal{R}_{2x}\otimes
\left(\mathcal{A}_{2x+1}\otimes \mathcal{A}_{2x+2}\right).$ It is shown in
\cite{gross:index-theorem} that $\mathcal{R}_{2x}$ is isomorphic to a matrix algebra $M_{s}$ and the index is defined as
\begin{equation}\label{eq:original-index}
\text{Original index} = \frac{s}{d}.
\end{equation}
It is clear from the standard form of the MPU (\ref{StandardForm}) that $\mathcal R_{2x}= v^{\dagger}( \Id\otimes M_\ell )v$ and hence
$${\rm ind}= -2\log_2\frac{\ell}{d}=-2\log_2(\text{Original index})\; .$$

% ============================================================

\end{document}